\documentclass[12pt]{article}

\usepackage{fullpage}
\usepackage[T1]{fontenc}
\usepackage{ae,aecompl}
\usepackage[dvips]{graphicx}
\usepackage{amssymb}
\usepackage{amsmath}
\usepackage[round,authoryear]{natbib}
\usepackage{color}
\usepackage{array}

\definecolor{webgreen}{rgb}{0,0.4,0}
\definecolor{webbrown}{rgb}{0.6,0,0}
\definecolor{purple}{rgb}{0.5,0,0.25}
\definecolor{darkblue}{rgb}{0,0,0.7}
\definecolor{darkred}{rgb}{0.7,0,0}
\definecolor{darkgreen}{rgb}{0,0.7,0}
\usepackage[pdfborder=false]{hyperref}
\hypersetup{colorlinks,citecolor=darkred,filecolor=black,linkcolor=darkblue,urlcolor=webgreen,pdfpagemode=none,
	pdfstartview=FitH}
\newcommand{\ignore}[1]{}
\newtheorem{lemma}{{\sc Lemma}}[section]
\newtheorem{prop}{{\sc Proposition}}[section]
\newtheorem{cor}{{\sc Corollary}}[section]
\newtheorem{theorem}{{\sc Theorem}}[section]
\newtheorem{defn}{{\sc Definition}}[section]

\newtheorem{example}{{\sc Example}}[section]

\newtheorem{remark}{{\sc Remark}}

\newenvironment{proof}{\noindent {\bf \sl Proof\/}:\enspace}
{\hfill $\blacksquare{}$ \vspace{12pt}}
\usepackage{sectsty}
\sectionfont{\centering\normalfont\scshape}
\subsectionfont{\centering\normalfont}

\begin{document}

\title{\bf Intelligent Machines and Incomplete Information}
\author{Sujata Goala,\thanks{Department of Mathematics, Dibrugarh University, Dibrugarh, India. email:\url{sujatagoala10@gmail.com}, Orchid Id \url{0009-0000-6149-5500}}~~Mridu Prabal Goswami\thanks{Indian Statistical Institute, Tezpur, India. email:\url{mridu@isine.ac.in}, Orcid Id \url{0000-0002-9131-3058}}~~and~~Surajit Borkotokey\thanks{Department of Mathematics, Dibrugarh University, Dibrugarh, India. email: \url{sborkotokey@dibru.ac.in}, Orcid Id \url{0000-0001-8447-4403}.~~\\We are extremely grateful to Trideep Baruah, Parashjyoti Borah, Nilotpal Chakraborty, Hari K. Choudhury,  Subhasish Dhal, Ratul Lahkar, Amarjyoti Mahanta, Debasis Mishra, Rupayan Pal, Abinash Panda, Soumendu Sarkar, Arunava Sen, Tridib Sharma, Phrangboklang Lyngton Thangkhiew, Rohit Tripathi,  and all the seminar participants, Department of Computer Science and Engineering, IIIT Guwahati, for their comments and suggestions. Errors and omissions are ours.  }}

\maketitle

\begin{abstract} \noindent         
	   We consider a firm that wants to hire two individuals who potentially differ in their efficiency levels, and may put in distinct levels of efforts. The firm is not informed about the levels of efficiency of the potential employees a priori. The classical principal-agent models assume that employees know their efficiency levels. Hence, these models design incentive-compatible mechanisms. We do not assume that employees know how efficient they are. Instead, we assume that the production technology of the firm is intelligent, and the machine-output reveals the efficiency levels of the employees. We call such a production technology to be an intelligent machine. The machine is modeled as a game in characteristic form, and thus the payoffs of the employees are determined by the Shapley shares, and the costs of efforts. Thus, we embed a cooperative game in a non-cooperative game of incomplete information and contribute to the literature on cooperative games under uncertainty. A characterization of ex-ante Nash equilibrium is established. We discuss our notion of intelligent machine in the context of explainable artificial intelligence in Remark \ref{remark:EAI}.                    
	   \end{abstract}

\section{\bf Introduction}
\label{sec:intro}  
Machine operator efficiency is defined to be the performance of an employee who operates industrial machinery. The performance of an employee is decided by her effort and efficiency/talent.
The importance of machine operating efficiency, or more broadly the importance of the interface between man and machine for machine productivity,  can hardly be overemphasized, see \citep{Wilson} for an elaborate discussion.
In this paper, a machine refers to a technologically sophisticated plant/factory/firm /a machine in the standard engineering sense such that the efficiency of the employees is of critical importance for production. 
Suppose a software firm requires its employees to know category theory.
The machine in this example refers to the software firm itself. Distinct levels of understanding of category theory refer to distinct levels of efficiency of the employees. Further, the distinct number of working hours is distinct levels of effort. We call such an abstract machine a {\bf Task Aggregator Machine (TAM)}, i.e., a machine that takes efforts and efficiency levels as inputs
to produce some output.{\footnote{ 
Software designed by the firm is its outputs. We assume that the outputs of a firm can be measured by real numbers. 
The outputs of a software firm can be measured by their market value.   
}}

The software firm mentioned above is an example of a TAM. Firms that produce self-driven cars are another class of examples. Smart engineers are required to produce self-driven cars, the latter can be called advanced AI models.{\footnote{We have used the word advanced only in a suggestive sense, and not in a sense in which it implies any continuity, the issues that arise while using it in the latter sense are pointed out in  \citep{Mitchell}.}} The skills required by an engineer who works for a firm that produces self-driven cars are advanced and technical in nature. We do not consider farms  and low-tech firms as examples of TAMs. Therefore, in our model the objective of the firm is not to maximize surplus. 
To elaborate on this remark we recall that the objective of the principal in the classic principal-agent models in economic theory is to maximize surplus, i.e., the principal/employer maximizes profit by paying wages to the agents/employees that keep agents’ payoffs at the lowest possible levels. 
In particular, the first best solution to the maximization problem minimizes wage bills, see Chapter $14$ in  \citep{Mas}.
Instead of maximizing surplus, modern technological firms share revenues with their employees. Hence, for these firms the marginal contribution of an employee to the firm’s output is an important consideration, and not just the effort level of an employee. 
The objective of the firm in our model is to find employees who can use the technology of the firm in such a way that entails the maximum revenue. 
The salary/remuneration of an employee in our model is based on the marginal contribution of the employee to the total revenue. A more efficient employee’s marginal contribution is higher than a less efficient employee.
Thus, the firm in our model is not looking to minimize its wage bill. Minimizing the wage bill is not the appropriate objective for technologically advanced firms since such firms require innovations from the employees, see \citep{Anderson} for a detailed discussion.

We assume that the firm hires two employees. An employee’s efficiency levels are either high or low. An employee’s effort level is either high or low. For every effort and efficiency vector TAM induces a cooperative game, i.e., a game in the characteristic form, see Definition \ref{defn:char_game}.           
Although the total revenue is generated by the employees jointly, the employees put efforts strategically. The strategic behavior of the employees entails a non-cooperative game among the employees. The firm management is not informed about the efficiency levels of the prospective employees, i.e., a situation of incomplete information. The firm management knows the probability distribution from which the efficiency levels are drawn. The payoff function of an employee has two parts, the benefit from efforts and the  cost of efforts. The benefits of the employees come in the form of shares of the total revenue. 
A pure strategy of an employee is a function from the set of possible efficiency levels to the set of possible efforts. We provide a characterization of symmetric Ex-ante Nash equilibria in pure strategies. The three Nash equilibria outcomes that we obtain are (a) all employees put in low efforts irrespective of their levels of efficiency (b) all employees put in high efforts irrespective of their levels of efficiency (c) the relatively more efficient employee puts high effort, and the less efficient employee puts low effort. Our characterization results provide insights into how probability distributions of levels of efficiency are related to equilibria. In particular, given a probability distribution of levels of efficiency, our results can tell us which strategies form equilibria, and which do not.  
Each equilibrium associates every efficiency profile of the employees to a level of output, thus   
each equilibrium entails a distribution on the set of outputs.
If we interpret a probability distribution of the levels of efficiency as a distribution of talents in the economy, then we can can say that our characterization results formalize the relationship between firm outputs and the distribution of talents. 
Distribution of talents is an important decision variable of a firm while making  
hiring decisions.

In Example \ref{ex:paradox} a distribution of talents gives rise to two ex-ante Nash equilibria. In one equilibrium the less efficient employee puts in low effort, and in the other equilibrium, she puts in high effort. It is difficult to predict which equilibrium will be played in such situations. To address this issue of multiple equilibria we consider {\bf rationalizable} equilibrium strategies. This notion pins down the distribution of talents to unique equilibria.
An important aspect of a game in characteristic form is its sets of all singleton coalitions. A singleton coalition is a situation in which an employee works alone. We interpret the singleton coalition that corresponds to an employee as the training or the probation period of that employee. Two employees that we consider in our model are the ones who survive their probation periods. Naturally, the employees who are fired after their probation periods do not have any marginal contribution to the firm and thus are not relevant for our model.
However, an individual's efforts for the firm when she works alone can differ from her efforts when she works jointly with another individual. 
She may do this to take advantage of the freeriding.  
We analyze Nash equilibria concerning this situation in Section \ref{sec:coalition_strategic}.

In our model TAM’s output reveals the effort and the level of efficiency of each employee. This revelation does not depend on employees' beliefs about their levels of efficiency. Therefore, we call such a firm an {\bf intelligent TAM}, see Definition \ref{defn:int} for a formal exposition. One of the main reasons for adopting this approach is that we wish to explore the implications of incorporating artificial intelligence for equilibrium outcomes in games with incomplete information. On the contrary classic principal-agent models that analyze surplus maximizing wage-efforts contracts assume that employees know how efficient they are. 
Thus, the optimization problems in these classes of models consider incentive constraints. An incentive constraint requires that the payoff obtained by an employee by pretending is not larger than the payoff obtained by behaving according to her true level of efficiency, see  \citep{Mas} or \citep{Laffont}.
However, the assumption that employees know how efficient they are may not be considered realistic. For example, it is not obvious that an individual who is trained in category theory knows exactly how good she is in the subject.{\footnote{\citep{Vazire} find self-perception of personality to be far from accurate. Further \citep{Parks} observe that people seem to not understand incentive compatibility well.}} In such situations, studying  strategic behavior based on the notion of incentive compatibility may not be considered appropriate.
In our model, we do not assume that employees know about their efficiency levels. In our model, the revelation of the efficiency levels is done by the intelligent TAM. Thus, our model provides a novel way to approach the problem of hiring in the presence of incomplete information by incorporating intelligent technology instead of taking the approach that depends on the reported information of the employees about their levels of efficiency.
 
The organization of the paper is as follows. 
In Section \ref{sec:TAM} we explain how we model TAMs as games in characteristic forms. We provide a review of the related literature in Section \ref{sec:lit}.   
In Section \ref{sec:model_results} we discuss the technical conditions imposed on the TAM, and the cost function of the employees.
In this section we also explain how we incorporate the probability distribution of efficiency levels into our model so that we can analyze the resulting economic environment by using the mathematics of games of incomplete information.In Section \ref{sec:char} we provide our main characterization result.  
In Section \ref{sec:coalition_strategic} we extend our model that incorporates strategic behavior that depends on the coalition that she is part of.  
We make concluding remarks in Section \ref{sec:con}.

\section{\bf Preliminaries: Relating TAMs with Cooperative Games }  
\label{sec:TAM}
We denote a TAM by $M$. We assume the number of employees to be two, and
the set of employees is denoted by $N=\{1,2\}$. Let $\mathbb{R}_{+}$ denote the set of non-negative real numbers.
Let $t_l$ and $t_h$ be two possible efficiency levels. We let $0$ denote 
the efficiency level of an individual who is not hired.{\footnote{This notation is introduced for the convenience of exposition.}}Thus, $T=\{t_{l},t_{h},0\}$ is the set of efficiency levels, where $t_l$ denotes the lower level of efficiency, i.e., low type; and $t_h$ denotes the highest level of efficiency, i.e., high type.
Formally, the notion of low and high type is expressed by an order $<_{t}$ on $\{t_{l},t_{h}\}$ which is $t_l <_{t} t_h$. 
Analogously,  
$E=\{e_l,e_h,0 \}$ denotes the set of possible levels of efforts with $e_l<_{e} e_{h}$; 
$e_{l}$ denotes the lower and $e_{h}$ the highest level of efforts, we do not assume these levels come from the set of numbers. 
We let $0$ denote the absence of effort of individual $i$ as an employee. Let $t_{i}$ denote a generic type of an individual. 
A type profile is denoted by $( t_1, t_2)$, and an efforts profile is denoted by $( e_1, e_2)$. 
That is, while writing a profile we write the corresponding entry for individual $1$ first and then for individual $2$.
A {\bf machine state} refers to an ordered list $( (e_1, e_2), (t_1, t_2)) $  where $e_1$ is the efforts of individual $1$ whose type is $t_{1}$. Analogously, $e_2$ denotes the effort of individual $2$ whose type is $t_{2}$.     
The output of the machine $M$ at the machine state $((e_1, e_2), (t_1, t_2))$ is $M(( e_1,  e_2),(t_1, t_2))$ and $M(( e_1,  e_2),(t_1, t_2))\in \mathbb{R}$, where $\mathbb{R}$ denotes the set of real numbers. 
We notice that $M(( e_1,  0),(t_1, t_{2})), t_2\neq 0$,  $M(( e_1,  e_2),(t_1, 0)), e_2\neq 0$, 
$M(( 0,e_2),(t_1, t_{2})), t_1\neq 0$, and  $M(( e_1,  e_2),(t_1, 0)), e_1\neq 0$ are not defined.
That is, if an individual is hired, then she cannot put zero effort; and if an individual is not hired, then speaking about her contributions to the firm is meaningless. Consider the machine state $((0,e_2),(0,t_{2}))$. In this state, individual $1$ is not hired and thus does not put effort. In this machine state, only individual $2$ is hired and she puts the efforts $e_{2}$, and $e_{2}$ must be either $e_{h}$ or $e_{l}$. Further, $t_{2}$ must be either $t_{h}$ or $t_{l}$. 
This machine state denotes the singleton coalition in which only individual $2$ is present  whose type is $t_{2}$ and effort is $e_{2}$. 
An analogous interpretation holds for the machine state $((e_1,0),(t_1,0))$.{\footnote{Often software firms give online coding tests to prospective employees, and the individuals who are eventually hired are asked to work in teams. If we imagine such  firms to be TAMs, then the individual who takes the test can be thought of as a singleton coalition.}} Further, $((0,0),(0,0))$ refers to the empty coalition, we set $M((0,0),(0,0))=0$, i.e., TAM cannot produce anything by itself. That is, we assume a technology that requires humans to produce an output. The machine state $((e_1,e_2),(t_1,t_2))$ denotes the grand coalition in which both individuals are present and individual $i$'s effort is $e_i$, her type is $t_i$ and $e_i\neq 0, t_i\neq 0$. We assume that both individuals are hired so that the grand coalition is formed.
{\footnote{Since we wish to study outcomes due to cooperation, we assume that the firm hires two individuals.}}
We call the machine state {\bf grand machine state} in which both individuals are present. In the standard cooperative game theory a coalition is identified with a list  of individuals, we identify a coalition with a list of individual specific characteristics, namely efficiency and effort. In particular, corresponding to every grand machine state there is a characteristic form game. For example, if $((e_l,e_{h}),(t_l,t_{h}))$ is the grand machine state, then 
$((e_l,e_{h}),(t_l,t_{h}))\mapsto M((e_l,e_{h}),(t_l,t_{h}))$, $((e_l,0),(t_l,0))\mapsto M((e_l,0),(t_l,0))$, $((0,e_{h}),(0,t_{h}))\mapsto M((0,e_{h}),(0,t_{h}))$, and 
$((0,0),(0,0))\mapsto M((0,0),(0,0))$ define a characteristic form game. 
We formalize the notion of TAM in Definition \ref{defn:TAM}. Let  $\mathbb{A}=\{((e_1,0),(t_1,0)), ((0, e_2 ),(0,t_2)), ((0,0),(0,0)), ((e_1,e_2),(t_1,t_2))\mid e_i \in \{e_l,e_h\}, t_i \in \{t_l,t_h\}, i=1,2\}$ be the set of admissible coalitions.

\begin{defn}\rm A {\bf Task Aggregator Machine, TAM  } is a function $M:\mathbb{A}\rightarrow \mathbb{R}_{+}$.  
\label{defn:TAM}	
\end{defn}

\noindent We define a game in characteristic form corresponding to a grand machine state as follows.

\begin{defn} \rm Given $M$ and the grand machine state $((e_1,e_2),(t_1,t_2))$, the restriction of $M$ to $\{((e_1,0),(t_1,0)), ((0, e_2 ),(0,t_2)), ((0,0),(0,0)), ((e_1,e_2),(t_1,t_2))\}$ is called a {\bf  Game in Characteristic Form}.

	\label{defn:char_game}	
\end{defn}

\noindent A key assumption of our paper is that $M$ is intelligent. A machine is intelligent if by observing the output of the machine the management can infer the effort-type combination that yields that output.
We now proceed to define an intelligent machine formally.

\begin{defn}\rm We call a TAM $M$ {\bf intelligent} if 
	
	\begin{enumerate}
		\item $M((e_1',e_{2}'),(t_{1}',t_{2}'))= M((e_{2}',e_{1}'),(t_{2}',t_{1}'))$. This condition can be interpreted as {\bf symmetry}. That is, output depends only on effort and efficiency and not on the identities of the employees.

		\item Except the {\bf symmetry} condition given above, $M((e_1',e_{2}'),(t_1',t_2'))\neq M((e_1'',e_{2}''),(t_1'',t_2''))$
		whenever $(e_1',e_{2}')\neq (e_1'',e_{2}'')$ or $(t_1',t_2')\neq (t_1'',t_2'')$.

	\end{enumerate}

	\label{defn:int}
\end{defn}	

\noindent The second property in the definition is the embodiment of the notion of an  intelligent TAM. It helps in computing the marginal contributions of the employees. Suppose the grand machine state is 
$((e_l,e_h),(t_l,t_h))$, and thus the output from the machine is $M((e_l,e_h),(t_l,t_h))$. By the second property, outputs for different grand machine states are different. Thus, the management can deduce the grand machine state that occurs. That is, the management can see that one of the employees is of type $t_l$ and has put efforts $e_l$, and the other one is of type $t_h$ and has put efforts $e_h$, but cannot decipher who is who from the output of the grand machine state. However, by the second property $M((e_l,0),(t_l,0))\neq M((0,e_h),(0,t_h))$. That is, the management has information about the outputs from the singleton coalitions. Now $M((e_l,e_h),(t_l,t_h))-M((e_l,0),(t_l,0))=$marginal contribution of employee $2$, with efficiency $t_h$ who puts in the effort $e_h$, to the output of the grand machine state when  employee $1$ with efficiency level $t_l$ puts in the effort $e_l$. 
Further, marginal contribution of employee $1=$  
$M((e_l,e_h),(t_l,t_h))-M((0,e_h),(0,t_h))\neq 0$. 
Since by the second property $M((0,e_h),(0,t_h))\neq M((e_l,0),(t_l,0))$, we obtain  
$$M((e_l,e_h),(t_l,t_h))-M((e_l,0),(t_l,0))\neq M((e_l,e_h),(t_l,t_h))-M((0,e_h),(0,t_h))$$. 

\noindent Thus, the management can deduce the marginal contributions of the employees. 
This enables the firm management to pay salaries on the basis of the marginal contributions of the employees.{\footnote{The chess engine Stockfish can differentiate between a ``bad'' move and a ``good'' move. Thus, it is an intelligent machine. However, Stockfish is not a TAM.}}
\citep{HartMoore} consider a two period scenario where in the first stage agents put efforts non-cooperatively and before the second stage all the uncertainties are revealed, i.e., before the second stage agents observe each others' actions taken in the first stage. In the second stage, gains are shared according to Shapley value. 
\citep{HartMoore} do not provide a mechanism that resolves the uncertainties.  
In our model TAM resolves the uncertainty. 
Therefore, the assumption in \citep{HartMoore} that information about efforts is symmetric among the employees is not required. 
This is an important point about our paper, i.e., we wish to understand the implications of intelligent machines for economic outcomes when agents do not know about each other. In the first stage of the two-player game in \citep{Nash} players declare their actions. The model in \citep{Nash} is abstract. However, as an example of the actions we can consider efforts. \citep{KoNey} consider a generalization of \citep{Nash} to the $n$ player scenario. \citep{Nash} and \citep{KoNey} define a value of strategic form games by using the notion of minmax payoff, and provide an axiomatic characterization of the value. In these models, players coordinate strategies, and hence they are models for cooperation in strategic games. 
In our model, employees do not coordinate efforts, and in particular they do not coordinate strategies since they may not know their types.
In \citep{Liu} players know their types, and have probabilistic beliefs about other players' types. 
To understand the main idea in \citep{Liu}, consider a player who is a member of a coalition. 
Now consider a counterfactual situation where she is contemplating deviating to another collective outcome. 
For such a deviation to take place other members of the coalition should also be willing to deviate, i.e., coordinate on the deviation. \citep{Liu} studies such coordination via consistency conditions imposed on the functions that map possible deviations to the set of beliefs on the type space. In other words, in \citep{Liu} coordination among the members in a coalition takes place in beliefs about whether they would collectively prefer an outcome, in \citep{KoNey} coordination takes place in strategies. Further, with the help of the consistency conditions pertaining to beliefs,  \citep{Liu} provide a characterization of stable outcomes, i.e., outcomes by deviating from which a coalition cannot benefit. The notion of stability that we consider is the Ex-ante Nash equilibrium. Our characterization results tell us how the probability distributions on the type space, i.e., distribution of types, is related to whether a strategy can be sustained as an Ex-ante Nash equilibrium or not.

The notion of an intelligent machine as discussed in \citep{Simkoff} puts our ideas into perspective.                        
\citep{Simkoff} describes machine intelligence as follows: ``machine intelligence by necessity involves deductive logic. For example, systems exhibiting true machine intelligence come to understand when
they’ve made mistakes, watch out for similar data that could lead to a
similar mistake the next time, and avoid doing so.'' 
By following the line of thinking as described in \citep{Simkoff} we consider TAMs that are capable of unambiguously deducing the efforts and the efficiency levels that are associated with the levels of outputs, hence TAMs are intelligent. In a low-tech firm, it is harder to make such deductions about efficiency. The reason for this is that the advanced technical skills of employees are not essential for these firms. Therefore, we do not consider them to be examples of TAMs.
Since the output of an intelligent $M$ reveals the effort and efficiency of the employees, we do not need to assume that employees know their efficiency levels. In the classical  models of Economics, it is assumed that employees know their efficiency levels, however the management does not. Such a situation is an example of a situation of asymmetric information. There can be situations where neither the employees know, nor the outputs reveal the contributions of the employees. Consider a situation where two individuals have written a paper jointly. Let us assume that the paper receives an award, and thus the two coauthors jointly receive prize money. The question now is how the two individuals should share the prize money. It is difficult for an individual to know her contribution to the paper. Further, from the outcome of the joint work, let the outcome be the prize money, it is difficult to deduce the contribution of each author of the paper. In other words, this is an example of a situation where neither the marginal contribution of an individual is computed from the joint output, nor do individuals know about their contributions. Thus, such situations are not examples of intelligent $M$, and we do not consider them in our analysis. 
We further note that such scenarios cannot be studied using the model in \citep{HartMoore} as well since \citep{HartMoore}  require symmetric information of efforts. The effort of an employee by itself is not enough for machine productivity. Even twelve hours of work every working day by an employee whose understanding of category theory is `not good’ may not be of significant importance to a software firm. Thus, we consider situations where it matters for a firm whether efforts come from a high-type or a low-type employee. Formally, situations that describe whether efforts come from high or low type are modeled by functions from  $\{t_l,t_h\}$ into $\{e_l,e_h\}$. We call such a function a {\bf strategy}. Let
$s_{i}$ denote the strategy of employee $i$. Then $(s_1,s_2)$ is called a {\bf strategy profile}, and the profile is called {\bf symmetric} if $s_1=s_2$.  
For any type profile  $(t_1,t_2)\in \{t_l,t_h\}\times \{t_l,t_h\} $, a strategy profile defines the grand machine state $((s_{1}(t_{1}),s_{2}(t_{2})), (t_1,t_2))$, and thus defines a characteristic form game. Given a strategy-profile $(s_1,s_2)$ we have the following collection of characteristic form games: 

\begin{enumerate}
	\item $M((s_{1}(t_l),s_{2}(t_h)),(t_l,t_{h}))$, $M((s_{1}(t_l),0),(t_l,0))$,
	$M((0,s_{2}(t_h)),(0,t_{h}))$, $M((0,0),(0,0))$ 
	
	\item $M((s_{1}(t_h),s_{2}(t_l)),(t_h,t_{l}))$, $M((s_{1}(t_h),0),(t_h,0))$,
	$M((0,s_{2}(t_l)),(0,t_{l}))$, $M((0,0),(0,0))$	
	
	\item $M((s_{1}(t_h),s_{2}(t_h)),(t_h,t_{h}))$, $M((s_{1}(t_h),0),(t_h,0))$,
	$M((0,s_{2}(t_h)),(0,t_{h}))$, $M((0,0),(0,0))$	
	
	\item $M((s_{1}(t_l),s_{2}(t_l)),(t_l,t_{l}))$, $M((s_{1}(t_l),0),(t_l,0))$,
	$M((0,s_{2}(t_l)),(0,t_{l}))$, $M((0,0),(0,0)).$ 
	\end{enumerate}

\noindent We observe that for a given strategy profile the characteristic form games depend only on the types. We assume that the management is uninformed about the efficiency levels of the job seekers. Thus, given a strategy profile, the uncertainty over the set of characteristic form games is the same as the uncertainty over the space $\{t_l,t_h\}$. We assume that the management knows the probability distribution of the type space  $\{t_l,t_h\}$.
A well-known interpretation of a strategy  $s_i$ is that nature reveals the type $t_i$ to employee $i$, and then $i$ decides the level of efforts as suggested by $s_i$ which is $s_i(t_i)$.
According to this interpretation $s_i$ is a conscious contingent plan of actions of $i$. However, we do not assume that employee $i$ knows her type.
{\footnote{Alternatively, we may relax the assumption that $i$ knows her type, and assume that $i$ holds a belief about her own type, and she knows what she truly believes.  
Such weakening from knowing to having a belief about 
being $t_l$ or $t_h$ is also not required. Since  $M$ is intelligent, eventually types will be revealed.}}
Since we do not assume that an employee knows her type, it implies that she does not know her strategy. We interpret  {\bf $s_i$ as a strategy that the nature plays through individual $i$}. Whether employee $i$ is aware of nature's move or not is not important for our analysis. An employee may believe that her type is $t_h$, and therefore puts in efforts $e_h$.
However, it may turn out that her true type is $t_l$. We make further comments on
our interpretation of a strategy in Remark \ref{remark:Shapley_fair}. 
In Section  \ref{sec:coalition_strategic} we consider the situation where an employee's effort depends on whether she works alone or in a team. 
Each of the four characteristic form games entails Shapley shares for each employee.
The Shapely share of employee $i$ is the average of the marginal contributions. Each employee contributes to two coalitions, the empty coalition and the grand machine state. Let the grand machine state be $((e_l,e_h),(t_l,t_h))$. Then the Shapely share of employee $1$  is $$Sh_{1}((e_l,e_h),(t_l,t_h))=\frac{M((e_l,e_h),(t_l,t_h))-M((0,e_h),(0,t_h))+M((e_l,0),(t_l,0))}{2},$$ where $M((e_l,0),(t_l,0))$ is employee $1$'s marginal contribution to the empty coalition, and $M((e_l,e_h),(t_l,t_h))-M((0,e_h),(0,t_h))$ is employee $1$'s marginal contribution when $1$ joins employee $2$ whose type is $t_h$ and puts effort $e_h$. We interpret the singleton coalition as the probation period of an employee. A analogues formula hod for employee $2$, see subsection \ref{sec:GIF}. 
In our model, an employee receives her Shapely share as her salary. 
\begin{remark}\rm The Shapley value is efficient, i.e., for all $(e_1,e_2)\in \{e_l,e_h\}\times \{e_l,e_h\}$ and $(t_1,t_2)\in \{t_l,t_h\}\times \{t_l,t_h\}$,  $Sh_{1}((e_1,e_2),(t_1,t_2))+
	Sh_{2}((e_1,e_2),(t_1,t_2))=M((e_1,e_2),(t_1,t_2))$. 
	The management or the investors in a company keep the major share of the profits/revenue to themselves. To include this fact, we can generalize Shapley shares such that the sum of the shares of the employees is less than the total revenue.  
	This only complicates our computations, and does not lead to any qualitative change in the results. 
\end{remark}

\noindent In recent times, cooperative game theory, and in particular Shapely Value, has found an important place in the explainable Artificial Intelligence research, see
\citep{LipCon}, \citep{Datta}, \citep{BoLu} and \citep{CovSu} for details. An important method based on Shapley Value which is used for assessing the contributions of an individual component in the total output in explainable AI is  Shapley Additive Explanations (SHAP in short). SHAP is proposed by \citep{LuSu}. One of the important motivations for using SHAP, as argued in \citep{LuSu}, is that SHAP value estimations are better aligned
with human intuitions. In the next remark we briefly discuss SHAP and intelligent TAM. 

\begin{remark}\rm In SHAP the inputs that go into a `predictive model' are features. In our notation the predictive model is denoted by $M$. Like the predictive models, $M$s also take inputs as features, note that efforts and types are features. Typically, a predictive model is considered to be a random variable.  
The composition of $M$ with a strategy profile, we are referring to the four characteristic form games described earlier, is a random variable that depends solely on the features $\{t_l,t_h\}$. A strategy profile and $\{t_l,t_h\}\times \{t_l,t_h\}$ generate four samples of features, where a typical sample is denoted by $( ( s_1(t_1), s_2(t_2)), ( t_1,t_2))$. Instead of the conditional expectation of the predictive model, in our model at each sample of features the intelligent machine entails a characteristic form game.             
\label{remark:EAI}
\end{remark}

\noindent Although the firm shares the revenue with the employees, the non-cooperative side of our model has important implications. 
For example, it is possible that if employee $2$ puts in high efforts, then the best response of employee $1$ is to put in low efforts. 
Alternatively, employee $1$ may put in high efforts to increase her Shapley share. However, high efforts coming from a less efficient employee may not be good for the firm's net welfare. Example  \ref{ex:paradox} in Section \ref{sec:char} makes this point, formally.

\section{\bf Relation to the Literature} 
\label{sec:lit}
This section reviews some more literature. Companies are increasingly using artificial intelligence based algorithms when hiring employees, see  \citep{Morris}. 
These algorithms analyze resumes of the job applicants to measure their personalities. Hence, these algorithms are examples of indirect mechanisms whose objective is to elicit private information about various features of the applicants. 
These algorithms are prediction machines, and the notion of a prediction machine is well discussed in  \citep{Agrawal}.  
\citep{D'Arinzo} find these algorithms to be unreliable, for instance, ``recruiters do not know why certain candidates are on page one of the ranking, or why certain people are on page ten of the ranking when they search for candidates", see  \citep{Morris}. 
In other words, intelligent machines that are being used currently to screen job seekers are in their infancy. 
\citep{Marwala}  imagine intelligent machines that can moderate agent-behavior to be in line with the expected behavior. 
Next, we discuss some literature related to solutions concepts of cooperative game theory with uncertainty.               

Fix a type profile $(t_1,t_2)$, and consider all possible effort profiles. This type profile entails a game in non-transferable utility (NTU), which is described as follows. Let  $V(S)$ denote a coalitional function, where $S$ is a non-empty subset of $\{1,2\}$. Then $V(\{1,2\})=\Big\{\Big( Sh_{1}((e_{1}',e_{2}'),(t_1,t_{2}))-C(e_{1}',t_{1}), Sh_{2}((e_{1}',e_{2}'),(t_1,t_{2}))-C(e_{2}',t_{2})\Big)\mid (e_{1}',e_{2}')\in \{e_{l},e_{h}\}\times \{e_{l},e_{h}\} \Big \}\subseteq \mathbb{R}^{2}$. 
Then $V(\{1\})=\{M((e_{1}',0), (t_{1},0))-C(e_{1}',t_{1})\mid e_{1}'\in \{e_l,e_h\}\}$. Further,  $$V(\{2\})=\{M((0,e_{2}'), (0,t_{2}))-C(e_{2}',t_{2})\mid e_{2}'\in \{e_l,e_h\}\}\subseteq \mathbb{R},$$ for details on NTU games see \citep{Hart}. 
An example in Section $1$ in \url{https://drive.google.com/file/d/1M5P3AfTVUpJvhQc2Pj8FL6GgXnUkz51b/view?usp=sharing} demonstrates that the NTU Shapley value payoff vector is different from the Nash equilibrium payoff vector of the non-cooperative game where the set of actions of the players is   $\{e_l,e_h\}$ and the payoff vectors of the non-cooperative game are given by the elements from the set  $V(\{1,2\})$.{\footnote{We do not need assumptions such as closed or convex on the set of payoff vectors of coalitions while applying the algorithm to compute NTU value  described in \citep{Hart}. .}}Hence, the solution concept discussed in this paper is an extension of Shapley value for Transferable Utility (TU) Games to Nash equilibrium which is a solution concept applied to non-cooperative games. Usually, TU solutions are extended to NTU solutions, see \citep{Hart}. 
Further, once we allow for the type profiles to be probabilistic, then we extend our solution concept from a TU game solution concept to a solution concept for games with incomplete information i.e., ex-ante Nash equilibrium. In \citep{Myerson} strategies are functions from type spaces to type spaces and they are contingent plans of actions of the players. Since we do not assume that employees know their types, we do not consider direct mechanisms and hence we  do not consider the kind of strategies that  \citep{Myerson} considers.         
\citep{Masuya} studies a model in which the worth of the singleton and grand coalitions are known and provides an axiomatic characterization of complete and superadditive extension of such games. \citep{CerGra} also consider 
cooperative games such that the worth of all the coalitions is not known.    
\citep{CerGra} provide a method to approximate standard solution concepts for games where the worth of all the coalitions is known.       
In our model, the worth of all the coalitions of the characteristic form games is known, however, which game is going to be played is not known.       
\citep{Pongou and Tondji} consider a model in which there are $n$ inputs that produce some output. The employees in our model are the analog of the inputs  \citep{Pongou and Tondji}.     
The quality of the inputs in \citep{Pongou and Tondji} is unknown. 
The level of efficiency can be thought of as the level of quality of an employee.    
\citep{Pongou and Tondji} characterize ex-ante, which \citep{Pongou and Tondji} call a priori, and Bayesian Shapley value.

\section{\bf Model and Results} 
\label{sec:model_results}
In this section, we discuss our model and results. First, we discuss the technology i.e., assumptions on $M$, and the cost function of the employees.    
\subsection{\bf Definitions and Assumptions on $M$ and $C$ }
\label{sec:defn}
To carry out our analysis we impose certain restrictions on $M$. The effort and type profiles are denoted by $e_{S}=(e_{i})_{i \in S}$, and $t_{S}=(t_{i})_{i \in S}$  respectively in which the $i^{th}$ coordinate is $0$, if $i\notin S$. A typical coalition is written as $((e_1,e_2),(t_1,t_2))$. We consider the following definitions.
\subsubsection{\bf Assumptions on $M$} 
We assume $M$ to be intelligent.  
{\bf Definition of ordering on type-coalition:} For all pairs of type profiles $(t_1,t_2)$ and $(t_{1}^{\prime},t_2^{\prime})$ we say 
$(t_1,t_2) <_{tt} (t_{1}^{\prime},t_2^{\prime})$ if and only if, 
  $t_{i} <_t t_{i}^{\prime}$ for all $i$ or if $ t_{i}= t_{i}^{\prime}$ for some $i$ then $t_{j} <_{t} t_{j}^{\prime}$ for $i\neq j$.   
{\bf Definition of ordering on effort-coalition:} For all pairs of effort  profiles $ (e_1,e_2) <_{ee} (e_{1}^{\prime},e_2^{\prime})$
if and only if, $e_{i} <_e e_{i}^{\prime}$ for all $i$ or if $ e_{i}= e_{i}^{\prime}$ for some $i$ then $e_{j} <_{e} e_{j}^{\prime}$ for $i\neq j$.
We assume $M$ to satisfy {\bf Monotonicity within types i.e., higher effort profile generates more return}, which is defined as follows.  
 Let $(t_1,t_2)$ be a type profile and two effort profiles 
 $(e^{\prime}_{1},e^{\prime}_{2})$ and $ (e_{ 1},e_2)$ with  $(e^{\prime}_{1},e^{\prime}_{2})<_{ee} (e_1,e_2)$
  then, $M((e^{\prime}_{1},e^{\prime}_{2}),(t_1,t_2)) < M((e_1,e_2),(t_1,t_2))$.    
We assume $M$ to satisfy {\bf Monotonicity within efforts i.e., more efficient type profile generates more return} which is defined as follows. 
 Let $(e_1,e_2)$ be an effort profile and two type profiles 
 $(t^{\prime}_{ 1},t^{\prime}_{2})$ and $ (t_1,t_2)$ with $(t^{\prime}_{1},t^{\prime}_{2}) <_{tt} (t_1,t_2)$
  then, $M((e_1,e_2),(t^{\prime}_{1},t^{\prime}_{2})) < M((e_1,e_2),(t_1,t_2))$.
We assume $M$ to satisfy 
{\bf Super-modularity i.e., increments in return for more efficient types are larger} which is defined as follows. For all two type profiles $t^{'}$, $t^{''}$ and two effort profiles, $e^{'}$, $e^{''}$ with  $t^{'}<_{tt} t^{''}$ and  $e^{'} <_{ee} e^{''}$, $  M( e^{''},t^{'})- M(e^{'},t^{'}) <  M( e^{''},t^{''})- M( e^{'},t^{''})$. As an example let $(t^{'},0)<_{tt}(t^{''},0)$, and $(e^{'},0)<_{ee} (e^{''},0)$. Super-modularity implies
$ M( (e^{''},0), (t^{'},0))- M( (e^{'},0), (t^{'},0)) <  M((e^{''},0), (t^{''},0))- M( (e^{'},0),(t^{''},0))$.

\subsubsection{\bf Assumption on $C$}
While putting effort each individual incurs a cost that depends on the type of the individual. Let $e\in E$ and $t\in T$ then the dependency of cost on effort and type is denoted by $C(e,t)$, and the cost function $C$ admits the following properties. As {analogous} to $M$ we do not define $C(e,t)$ if either one of them is $0$. Thus when we say $C$ to be function we mean that $C$ maps $(e,t)$ to real a number when neither $e$ nor $t$ is $0$.       
We assume $C$ to satisfy {\bf monotonicity of cost in efforts i.e., higher levels of efforts cost more. } which is defined as follows. For all  $t \in \{ t_l,t_h \}$,  $C(e_l,t) < C(e_h,t)$ where  $e_l <_{e} e_h$.
We assume $C$ to satisfy  {\bf efficiency in type i.e., higher type incurs lower cost} which is defined as follows. For all  $e \in \{ e_l,e_h \}$, $C(e,t_h) < C(e,t_l)$ where  $t_l <_{t} t_h$. We assume $C$ to satisfy {\bf Sub-modularity i.e., increment in cost for the efficient type is lower} which is defined as follows. For the pair of  $t_l$, $t_h$ with $t_l <_{t} t_h$ and $e_l$, $e_h$  with $e_l <_{e} e_h,$
$C(e_h,t_h) - C(e_l,t_h) < C(e_h,t_l) - C(e_l,t_l)$. In our study we can include the case, $C(e_h,t_h) - C(e_l,t_h) < C(e_h,t_l) - C(e_l,t_l)$, this condition incorporate the case : Cost function  is zero and then we can go back to the Shapley value as a particular case.
\medskip
\noindent The game of incomplete information is discussed next.

\subsection{\bf The Game of Incomplete Information}
\label{sec:GIF}
In this section we describe the game of incomplete information that we consider.    
Let $\Gamma = \{ N,T^{'2}, E^{'2}, p, G, (g_t)_{t \in T^{'}{^{2}}} \}$ denote the game of incomplete information whose components are defined as follows:
$N=\{1,2\}$ denote the set of players, $T^{'2}=T'\times T'$, $E^{'2}=E'\times E'$, where $T'=\{t_l,t_h\}, E'=\{e_l,e_h\}$. 
We assume that players' types are drawn independently and identically according to the probability distribution $p$, and $p(t_{l})>0, p(t_{h})>0$ with $p(t_{l})+ p(t_{h})=1$. Let $\mathbb{P}=\{(p(t_h),p(t_l))\mid p(t_{l})>0, p(t_{h})>0, p(t_{l})+p(t_{h})=1\}$ denote the set of all probability distributions with support $\{t_l,t_h\}$.
The probability measure $p$ represents the belief of the management about the distribution of types in the economy from which employees are drawn.
For example, $p$ may be a relative frequency empirical distribution. The shape of this distribution may be an important factor that determines the high concentration of high tech firms in some geographical region. 
For example, \citep{Audre} and \citep{Fig} look at the effect of educational attainments of the labor force on the firm location decision.{\footnote{Since efficiency is latent, educational attainment may not be a good proxy for efficiency. 
}}
Since we assume that both individuals are hired, 
the grand coalition is formed. Consider a grand machine state $((e_1,e_2),(t_1,t_2))$. 
By definition of a grand  machine state $e_i\in \{e_l,e_h\}$ and $t_i\in \{t_l,t_h\}$ for $i=1,2$. Thus consider the game in characteristic form defined by $M((e_1,e_2),(t_1,t_2)), M((e_1,0),(t_1,0)), M((0,e_2),(0,t_2)), M((0,0),(0,0))$. 
Let the Shapley shares of this game be denoted by $Sh_{i}((e_1,e_2), (t_1,t_2))$, $i=1,2$. 
The Shapley share of employee $1$ is  $Sh_1((e_1,e_2),(t_1,t_2))= \frac{M((e_1,e_2),(t_1,t_2))-M((0,e_2),(0,t_2))+M((e_1,0),(t_1,0))}{2}
$, and that of employee $2$ is  
$Sh_2((e_1,e_2),(t_1,t_2))= \frac{M((e_1,e_2),(t_1,t_2))-M((e_1,0),(t_1,0))+M((0,e_2),(0,t_2))}{2}
$. The payoff of employee $1$ at $(e_{1},e_{2})\in \{e_l,e_h\}\times \{e_l,e_h\}$ and $(t_{1},t_{2})\in \{t_l,t_h\}\times \{t_l,t_h\}$ is $Sh_{1}((e_1,e_2), (t_1,t_2))-C(e_1,t_{1})$, and that of employee $2$ is $Sh_{2}((e_1,e_2), (t_1,t_2))-C(e_2,t_{2})$.
Given a profile of efficiency levels, and employee $i$'s effort level, the value of the grand coalition depends on the effort of the other individual. Thus, the marginal contribution of $i$ is affected by the efforts of agent $j$. 
That is, for any fixed type profile $Sh_{i}((e_{1},e_{2}),(t_1,t_2))-C(e_1,t_1)$ depends on the effort level of $j$. Thus, we have a well defined stage game at each type profile as defined in according to  \citep{Maschler Solan and Zamir}.
$G$ is the set of stage games, $g_{t}$ refers to the stage game at the type profile $(t_1,t_2)\in \{t_{l},t_{h}\}\times \{t_{l},t_{h}\} $. 
Since type profiles are probabilistic, stage games are probabilistic as well. This defines a game  of incomplete information.  
We define pure strategy next. Since this is the only notion of strategy, we use in this paper from now on we call a pure strategy to be a strategy. A pure strategy or simply a strategy is a function from the set of types  $\{t_{l},t_{h}\}$ to the set of actions $\{e_{l},e_{h}\}$. There are four possible pure strategies defined below.

\begin{defn}\rm  A {\bf strategy (or a pure strategy) } $s_{e_1e_2}$  is a function $s_{e_1e_2}:\{t_l,t_h\}\rightarrow \{e_l,e_h\}$,  defined as \\

	$s_{e_1e_2}(t) = \left\lbrace 
	\begin{array}{cc}
		e_1 &\text{if}~ t=t_l\\
		e_2 &\text{if}~ t=t_h,\\
	\end{array}\right. 
	$\end{defn}

\noindent Given a strategy profile $(s_{1},s_{2})$  at $((s_{1}(t_{1}),s_{2}(t_{2})),(t_{1},t_{2}))$ payoff of player $1$ is given by $Sh_{i}((s_{1}(t_{1}),s_{2}(t_{2})),(t_{1},t_{2}))-C(s_{1}(t_{1}),t_{1})$. Likewise  $Sh_{2}((s_{1}(t_{1}),s_{2}(t_{2})),(t_{1},t_{2}))-C(s_{2}(t_{2}),t_{2})$ denotes the Shapley share of player $2$. 
The ex-ante expected payoff of employee $1$ for the play of the strategy-profile $(s_{1},s_{2})$ is $$\Pi_{1}(s_{1},s_{2})=  \sum_{ (t_{1},t_{2})\in T^{'2}   }  [Sh_{1}((s_{1}(t_{1}),s_{2}(t_{2})),(t_{1},t_{2}))-C(s_{1}(t_{1}),t_{1})]p(t_{1})p(t_{2})$$ and that for employee $2$ is 
$$\Pi_{2}(s_{1},s_{2})=  \sum_{ (t_{1},t_{2})\in T^{'2}}  [Sh_{2}((s_{1}(t_{1}),s_{2}(t_{2})),(t_{1},t_{2}))-C(s_{2}(t_{2}),t_{2})]p(t_{1})p(t_{2}).$$

\noindent Since the management knows the distribution of types only, and we do not assume that individuals know their types, ex-ante expected payoff is the appropriate notion of payoff for our model.
Our model has similarities with Biform games, the latter are two stage games  introduced in
\citep{BrandStuart}. 
Biform games in \citep{BrandStuart}
are games with complete information. Our model and Biform games are similar in the sense that in both models non-cooperation leads to cooperative games. That is,  each profile of nonoperative strategies leads to a cooperative game. 
The payoff of a player in a Biform game is determined by the value of the cooperative game. In particular, the payoff of a player is equal to a weighted average of core allocations. 
The payoff functions of the agents in our model have two components, one among them can be interpreted as cooperative, and the other non-cooperative.  
If we assume the cost of efforts to be zero in our model, only then our game can be classified as a Biform game. In the first stage of a Biform game, employees choose efforts without incurring any cost, and before the second stage employees observe their efforts as in \citep{HartMoore}.
But they do not observe the types of each other, and they face a probability distribution over a set of characteristic form games. Then at the end of the second stage expected weighted core allocations are obtained as payoffs. 
In our model choices of efforts have direct consequences for the payoffs in terms of cost of efforts, which is not the case in Biform games.
\citep{BrandStuart} mention it explicitly that the direct consequences of first stage moves are not payoffs.   
Instead of Biform games, our model is better interpreted as a non-cooperative game where the benefits of a player occur due to cooperation, and the cost occurs due to non-cooperation. 
Hence, our model embeds cooperation into a non-cooperative game. Let $S_{i}$ be the set of all strategies of individual $i$. The notion of ex-ante Nash equilibrium is defined below.

\begin{defn}\rm A strategy profile $(s_1^{*},s_2^{*})$ is an {\bf ex-ante  Nash equilibrium } for the game $\Gamma$ if and only if  
($i$) $\Pi_1(s_1^{*},s_2^{*}) \geq \Pi_1(s_1,s_2^{*})$ for all $s_1^{*} \neq s_1,~ s_1^{*},s_1 \in S_1$; and
 ($ii$) $\Pi_2(s_1^{*},s_2^{*}) \geq \Pi_2(s_1^{*},s_2)$ for all $s_2^{*} \neq s_2,~ s_2^{*},s_2 \in S_2$ hold. 
\end{defn}

\noindent An ex-ante Nash equilibrium is a notion of stability for the profiles of strategies. If $ (s_ {1} ^ {*}, s_ {2} ^ {*}) $ is such an equilibrium, then no individual $i$ is expected to be better off if $i$ does not end up behaving, consciously or unconsciously, according to $s_ {i} ^ {*} $ and $j$ behaves according to $s_ {j} ^ {*} $.

\begin{remark}\rm Since the strategy employee $i$ believes she plays need not be what she actually plays, the payoff that she believes she will obtain may be different from what she actually obtains. We may wonder if this discrepancy can create disincentives for employees and complaints against the management. As long as the process that generates salaries is transparent and fair, the discrepancy mentioned above is is unlikely to be a matter of serious concerns. An intelligent TAM by definition is transparent. Further, fairness axioms characterize 
the Shapley Value.
\label{remark:Shapley_fair}	
\end{remark}	

\noindent We define Symmetric ex ante Nash equilibrium next.

\begin{defn}\rm 
A strategy profile $(s_1^{*},s_2^{*})$ is a {\bf symmetric ex-ante Nash equilibrium (SNE)} for $\Gamma$ if and only if: $(i)$ $(s_1^{*},s_2^{*})$ is an ex-ante Nash equilibrium and
$(ii)$ $s_1^{*}=s_2^{*}$.
\end{defn}

\noindent We shall call the SNE $(s_1^*,s_2^*)$ as the SNE in the strategy $s_{e_1e_2}$ if $s_{e_1e_2}=s_i^*$ for $i=1,2$.  An SNE in the strategy $s_{e_he_l}$, i.e., $s_{e_he_l}(t_l)=e_h$, $s_{e_he_l}(t_h)=e_l$, does not exist in our model.    
We may interpret the other three symmetric equilibria as follows:
$s_{e_le_l}$ is the equilibrium in which TAM receives low efforts from all individuals irrespective of their efficiency levels; $s_{e_he_h}$ is the equilibrium in which TAM receives high efforts from all individuals irrespective of their efficiency levels; $s_{e_le_h}$ is the equilibrium in which TAM receives low efforts from an individual who is of low type and receives high efforts from a type who is of high type.
A firm may consider the distribution of efficiency levels as a factor before deciding on a location because equilibrium outcomes depend on the distributions.                
Therefore we provide a characterization of SNEs that provides a classification of SNEs by the probability distributions on $\{t_l,t_h\}$.      
In particular our computations provide information on what range of $p(t_h)$ gives what kind of equilibria. 
For example we show that symmetric equilibrium in $s_{e_le_l}$ exists if and only if $p$ takes values in an interval around $0$. Further, we identify the upper bound of the interval which depends on the Shapley value and the cost of efforts. 
It may be easier for a firm to take a decision about   
locating itself if for a given distribution the corresponding SNE is unique. The uniqueness pins down the possible equilibrium behavior uniquely.     
Thus we consider a notion of rationalizable equilibrium.

\begin{defn}\rm A strategy $s_{e_1e_2}$ is {\bf rationalizable} if there is a probability distribution $p$ that makes $s_{e_1e_2}$  a unique SNE.{\footnote{We note that the notion rationalizable strategy in our paper is very different from the rationalizable strategies defined in \citep{Bern}. The notion of rationalizability in \citep{Bern} has to do with strategies that are consistent with beliefs of players' about their strategies, in our paper the notion has to do with uniqueness of equilibrium and hence it is about equilibrium selection.}}
\end{defn}   
 
\noindent Since there is no SNE in $s_{e_he_l}$, this strategy is  not rationalizable. The other three equilibria are rationalizable.
Since different strategies entail different kinds of grand coalitions, rationalizable strategies provide information about the nature of the grand coalition that may form.  
The notion of rationalizability in 
\citep{Pongou and Tondji} is different from ours. In \citep{Pongou and Tondji} a player is an input to a production function, and a pure strategy of an input is quality. 
A mixed strategy of an input is defined to be a probability distribution on 
the set of pure strategies.
\citep{Pongou and Tondji} call a vector of mixed strategies rationalizable if the vector constitutes a Nash equilibrium of the associated complete information game. 
In the associated game the payoff from a vector of pure strategy is the Shapley value from the characteristic form game in which each player is identified with a quality of an input. We discuss our main results next.

\subsection{\bf Main Results: Characterizations of SNEs}
\label{sec:char}
We state our results in this section.             
For the sake of convenience of exposition instead of calling a game by $\Gamma$ we call it by $\Gamma_{p}$ since the only parameter that we vary while studying equilibria is $p$.
Further, we fix individual $2$'s strategy. For example, while studying $(s_{e_he_h},s_{e_he_h})$ as equilibrium we fix individual $2$'s strategy at $s_{e_he_h}$, and then argue that player $1$ cannot be made better off from deviation from the strategy that is assumed for player $2$. 
To study individual $1$ let 
$\Delta C_{t_1}= C(e_h,t_1)-C(e_l,t_1)\equiv$ increment in cost due to an increase in efforts at type $t_1$, and  
$\Delta_{e_1e_2}^{e_{1}^{\prime}e_{2}^{\prime}} Sh_1(t_1t_2)=Sh_1((e_{1}^{\prime},e_{2}^{\prime}),(t_1,t_2))-Sh_1((e_1,e_2),(t_1,t_2))\equiv$ 
change in Shapley share due to a change in efforts at the type profile $(t_1,t_2)$.

\begin{prop}\rm Let $M$  be super-modular and $C$ sub-modular: 
		$(i)$ $(s_{e_he_h},s_{e_he_h})$ is an SNE of $\Gamma_{p}$ for some $p \in \mathbb{P}$
		and  $(ii)$ $(s_{e_he_h},s_{e_he_h})$ is not an SNE of $\Gamma_{p^{\prime}}$ for some $p\neq p' \in \mathbb{P}$; 
	if and only if $\Delta_{e_le_h}^{e_he_h} Sh_1(t_lt_l) < \Delta C_{t_l} <  \Delta_{e_le_h}^{e_he_h} Sh_1(t_lt_h)$.  
	
	\label{prop:SNE_ineq_hh}
\end{prop}

\begin{proof} See the Appendix at the end.  
	
\end{proof}

\noindent Proposition \ref{prop:SNE_ineq_hh} identifies the particular pair of super-modular $M$ and sub-modular $C$ for which there is a probability distribution over $\{t_{h},t_{l}\}$ for which there is an SNE for the strategy $s_{e_{h}e_{h}}$. The next corollary gives a range of probabilities on $p(t_{h})$ for which we obtain $(s_{e_he_h},s_{e_he_h})$ as a symmetric ex-ante equilibrium. 

\begin{cor}\rm Let $M$ be super-modular and $C$ sub-modular. Then $(s_{e_he_h},s_{e_he_h})$ is an SNE of $\Gamma_{p}$ if and only if $p(t_h) \in [\frac{\Delta C_{t_l}-\Delta_{e_le_h}^{e_he_h} Sh_1(t_lt_l)}{\Delta_{e_le_h}^{e_he_h} Sh_1(t_lt_h)-\Delta_{e_le_h}^{e_he_h} Sh_1(t_lt_l)},1)$.
\label{Interval-hh}
\end{cor}
 
\begin{proof} See the Appendix at the end.  
	
\end{proof}\\
\noindent The next result is about $s_{e_{l}e_{l}}$.  

 \begin{prop}\rm \label{prop:SNE_ineq_ll} Let $M$ be super-modular and $C$ sub-modular:  
	$(i)$ $(s_{e_le_l},s_{e_le_l})$ is an SNE of $\Gamma_{p}$, for some $p \in \mathbb{P}$ ;
		$(ii)$ $(s_{e_le_l},s_{e_le_l})$ is not an SNE of $\Gamma_{p^{\prime}}$, for some $p^{\prime} \in \mathbb{P}$ 
if and only if $\Delta_{e_le_l}^{e_he_l} Sh_1(t_ht_l) < \Delta C_{t_h} <  \Delta_{e_le_l}^{e_he_l} Sh_1(t_ht_h)$.  
	\end{prop}
A proof similar to the  proof of Proposition \ref{prop:SNE_ineq_hh} can be found in Section $2$ at
 \url{https://drive.google.com/file/d/1M5P3AfTVUpJvhQc2Pj8FL6GgXnUkz51b/view?usp=sharing}.
The next corollary entails a range of probabilities on $p(t_{h})$ for which we obtain such symmetric ex-ante equilibria.

\begin{cor}\rm
	Let $M$  be super-modular and $C$ be sub-modular. Then $(s_{e_le_l},s_{e_le_l})$ is an symmetric ex-ante equilibrium of $\Gamma_{p}$ if and only if $p(t_h) \in (0,\frac{\Delta C_{t_h}-\Delta_{e_le_l}^{e_he_l} Sh_1(t_ht_l)}{\Delta_{e_le_l}^{e_he_l} Sh_1(t_ht_h)-\Delta_{e_le_l}^{e_he_l} Sh_1(t_ht_l)}]$.
	\label{Interval-ll}
\end{cor}
\begin{proof} A proof can be found in Section $2$ of\\ \url{https://drive.google.com/file/d/1M5P3AfTVUpJvhQc2Pj8FL6GgXnUkz51b/view?usp=sharing}.    
	
\end{proof}

\noindent The next result is about $(s_{e_l}, s_{e_h})$.  

\begin{prop}\rm Let  $M$ be super-modular and $C$ sub-modular: $(a)$ $(s_{e_le_h},s_{e_le_h})$ is an SNE of $\Gamma_{p}$, for some $p \in \mathbb{P}$;
		$(b)$  $(s_{e_le_h},s_{e_le_h})$ is not an SNE of $\Gamma_{p^{\prime}}$, for some $p^{\prime} \in \mathbb{P}$;
\medskip
	
	\noindent if and only if exactly one of the following holds:
	\begin{itemize}
		\item [$(i)$] at least one of the following holds   
		\begin{itemize}
			\item [$(a)$]
			$\Delta_{e_le_l}^{e_he_l} Sh_1(t_lt_l) < \Delta C_{t_l} <  \Delta_{e_le_h}^{e_he_h} Sh_1(t_lt_h)$,
			\item [$(b)$] $ \Delta_{e_le_l}^{e_he_l} Sh_1(t_ht_l) < \Delta C_{t_h} <  \Delta_{e_le_h}^{e_he_h} Sh_1(t_ht_h)$,
		\end{itemize}
		\item  [$(ii)$] at least one of the following holds
		
		\begin{itemize}
			\item [$(a)$] $\Delta_{e_le_h}^{e_he_h} Sh_1(t_lt_h) < \Delta C_{t_l} <  \Delta_{e_le_l}^{e_he_l} Sh_1(t_lt_l)$,
			\item [$(b)$]  $ \Delta_{e_le_h}^{e_he_h} Sh_1(t_ht_h) < \Delta C_{t_h} < \Delta_{e_le_l}^{e_he_l} Sh_1(t_ht_l) $,
		\end{itemize}
	\end{itemize}
	\label{prop:SNE_ineq_lh}
\end{prop}

\begin{proof} A proof of the result can be found in Section $3$ at \url{https://drive.google.com/file/d/1M5P3AfTVUpJvhQc2Pj8FL6GgXnUkz51b/view?usp=sharing}. 
\end{proof}	

\noindent The next corollary gives a range of probabilities on $p(t_{h})$ for which one can obtain such symmetric ex-ante equilibria.

\begin{cor}\rm
Let $M$  be super-modular and $C$ sub-modular. Then $(s_{e_le_h},s_{e_le_h})$ is an SNE of $\Gamma_{p}$ if and only if $p(t_h) \in [\frac{\Delta C_{t_h} -\Delta_{e_le_l}^{e_he_l} Sh_1(t_ht_l)}{\Delta_{e_le_h}^{e_he_h} Sh_1(t_ht_h)-\Delta_{e_le_l}^{e_he_l} Sh_1(t_ht_l)},\frac{\Delta C_{t_l} -\Delta_{e_le_l}^{e_he_l} Sh_1(t_lt_l)}{\Delta_{e_le_h}^{e_he_h} Sh_1(t_lt_h)-\Delta_{e_le_l}^{e_he_l} Sh_1(t_lt_l)}] \subseteq (0,1)$.
	\label{Interval-lh}
\end{cor}

\begin{proof} A proof of the result can be found in Section $3$ at \url{https://drive.google.com/file/d/1M5P3AfTVUpJvhQc2Pj8FL6GgXnUkz51b/view?usp=sharing}.    
\end{proof}	
\begin{prop}\rm Let $M$  be super-modular and $C$ sub-modular. There exists no $p$ for which $(s_{e_he_l},s_{e_he_l})$ is an SNE.  
\label{prop:no_nash}	
\end{prop}
\begin{proof} A proof of the result can be found in Section $4$ at \url{https://drive.google.com/file/d/1M5P3AfTVUpJvhQc2Pj8FL6GgXnUkz51b/view?usp=sharing}.
\end{proof}

\noindent 
Our characterization results lay down necessary and sufficient conditions for SNEs in terms of the parameters of our model. These conditions can be computed by using the parameters of the model. The three main propositions provide us with nontrivial intervals,  i.e., an interval that is neither a singleton set nor an empty set, for which SNEs in $s_{e_he_h}, s_{e_le_l}$ 
and $s_{e_le_h}$ exist. The three corollaries provide the ranges of these intervals. Our results also tell us when a particular equilibrium does not exist. If a probability distribution on  $\{t_l,t_h\}$ 
represents a distribution of efficiency, and if efficiency is interpreted as talent, then our results provide us with information about what kind of stable or equilibrium behavior outcomes may be expected if a distribution is given. The bounds on the intervals in the three corollaries are given by quantities that are functions of  $M$ and $C$. Since the pair $M, C$ define technology in our paper, our results provide an indirect mechanism to study observed behavior from the perspective of existing technology. As an example, consider $\frac{\Delta C_{t_h}-\Delta_{e_le_l}^{e_he_l} Sh_1(t_ht_l)}{\Delta_{e_le_l}^{e_he_l} Sh_1(t_ht_h)-\Delta_{e_le_l}^{e_he_l} Sh_1(t_ht_l)}$. This ratio can be interpreted as

$\frac{    \text{the surplus in the incremental cost of efforts over the change in the Shapley share of employee 1 due to an increase in her efficiency}}  {\text{the surplus in the change in the Shapely share of employee $1$ due to the change in the efficiency of employee $2$}}$ 

$\equiv \frac{\text{the net internal effects of change in efforts and efficiency in cost} }{\text{the net external effect of efficiency on benefits} }$

\noindent This ratio lies between $0$ and $1$. Thus, this ratio can also be interpreted as a price of being employed in a firm in which cooperation entails external benefits for the employee.       
Super-modularity of $M$ and Sub-modularity  of $C$, i.e., two important features of technology in our paper, play an important role in making the price lie between $0$ and $1$. The lower bound of the interval in Corollary \ref{Interval-ll} is smaller than the lower bound of the interval in Corollary \ref{Interval-lh}, and the lower bound of the interval in Corollary \ref{Interval-hh} is bigger than $0$. Thus, the strategies $s_{e_le_l}$ and $s_{e_he_h}$ are rationalizable.
However, if we assume $M$ to be concave then $s_{e_le_h}$ is also rationalizable. See Section $6$ at \url{https://drive.google.com/file/d/1M5P3AfTVUpJvhQc2Pj8FL6GgXnUkz51b/view?usp=sharing} for an example that shows that if $M$ is not concave then the strategy $s_{e_le_h}$ is not rationalizable. 
The notion of a concave $M$
is defined below.

{\begin{defn}\rm \label{def:concave} Consider a grand machine state.   
		M is said to be {\bf concave within the type profile $t^{'}$}  if for  three effort profiles, $e^{'}$, $e^{''}$, $e^{'''}$  with  $e^{'}<_{ee} e^{''}<_{ee} e^{'''}$, $   M( e^{'''},t^{'})- M( e^{''},t^{'}) < M( e^{''},t^{'})- M(e^{'},t^{'})$.  
	\end{defn} 

\noindent Concavity of $M$ says the increase in the output of $M$ is smaller at higher efforts. The following corollary provides a characterization of rationalizable strategies.

{\begin{cor}\rm Let $M$ be super-modular and $C$ sub-modular, then {$s_{e_le_l}$ and $s_{e_he_h}$} are rationalizable. If $M$ is also concave within type profile, then $s_{e_le_h}$ is also rationalizable. \label{rationalizable}  
\end{cor}
\begin{proof} The first part follows from the discussion above. If $M$ is concave, then the intervals in all the three corollaries above are non-empty, and they are pairwise disjoint. For details see the Appendix at the end. 
\end{proof}}

\medskip
\noindent It is possible that if we take the union of the intervals obtained in the corollaries above, then we may not obtain $[0,1]$. However, this should not be surprising because for certain games $\Gamma_{p}$ SNE may not exist since SNEs are pure strategies. If the intervals in corollaries \ref{Interval-hh} and \ref{Interval-lh} intersect, then we may wonder whether it is better for the firm that both employees put high effort according to the SNE in the strategy $s_{e_he_h}$. Example \ref{ex:paradox} demonstrates that it may not be so. First we define expected net welfare of the firm. 
The expected net welfare of the firm from SNE in the strategy $s_{e_le_h}$ is:  

\noindent $EW(s_{e_le_h},p)=2[p(t_l)p(t_l)\{ Sh_{1}((e_l,e_l),(t_l,t_l))-C(e_l,t_l)\}+p(t_l)p(t_h)\{ M((e_l,e_h),(t_l,t_h))-(C(e_l,t_l)+C(e_h,t_h))\}+p(t_h)p(t_h)\{Sh_{1} ((e_h,e_h),(t_h,t_h))-(C(e_h,t_h))\}]$

\noindent The expected welfare from SNE in the strategy $s_{e_he_h}$ is:

\noindent $EW(s_{e_he_h},p)= [p(t_l)p(t_l)\{ Sh_{1}((e_h,e_h),(t_l,t_l))-C(e_h,t_l)\}+p(t_l)p(t_h)\{ M((e_h,e_h),(t_l,t_h))-(C(e_h,t_l)+C(e_h,t_h))\}+p(t_h)p(t_h)\{Sh_{1} ((e_h,e_h),(t_h,t_h))-(C(e_h,t_h))\}]$

\noindent In Example \ref{ex:paradox} we construct a super-modular $M$ and sub-modular $C$, and show that the expected net welfare from $s_{e_le_h}$ is higher than that of $s_{e_he_h}$.

\begin{example}\rm All the tables related to this example are in the Appendix at the end. Table $1$ describes the TAM. For example, $7=M((e_l,0),(t_l,0))$. 
	That is, $7$ is the value of the singleton coalition when 
	only player $1$ is present in the coalition. Likewise $25$ is the value of the grand coalition when both players are of the type $t_h$, and put effort $e_h$. That is, $M((e_h,e_h),(t_h,t_h))=25$. Table $2$ describes a cost function. From Table $3$ we see that for $p(t_{h}) \in [0.578948, 0.794872]$, which is the intersection of the range of $p(t_h)$ in the second and the third row of Table $3$, expected welfare from the SNE in the strategy $s_{e_le_h}$ is higher than the SNE in $s_{e_he_h}$.          
	\label{ex:paradox} 
\end{example}

\noindent We end this section with a remark about extending our model to more than two players.  
  
\begin{remark}\rm If the number of players $n$ is greater than $2$,  then finding the range of probabilities for which a strategy is equilibrium require solving equations is $n-1$ degree. That is, the conceptual framework remains the same as the two employee scenario, but the procedure to find the range of probabilities is computationally more cumbersome when the number of employees is more than $2$.                
\end{remark}	

\noindent The extension of our main model to incorporate the strategic behavior across coalitions is discussed next.   

\subsubsection{\bf Strategic Behavior Across Coalitions} 
\label{sec:coalition_strategic}
In this section we consider the situation where an employee's efforts are different when she works on the machine jointly compared with the situation where she works alone. An employee's effort may be high during her probation period so that she can impress the management, and low when she works jointly with the other employee so that she can take advantage of free-riding.        
Let  $C_{i}=\{\{i\},\{i,j\}\}$ denote the set of coalitions that employer $i$ can be part of. Next we define a strategy of employee $i$ below.

\begin{defn}\rm  A {\bf strategy (or a pure strategy)} of $i$ is a collection of two functions $s_{i}^{\alpha}:\{t_l,t_h\}\rightarrow \{e_l,e_h\}$,  
$\alpha\in C_i$.  Here $s_{i}^{\alpha}$ denotes the strategy when $\alpha$ is the coalition.   
We denote a strategy of $i$ by $(s_i^{\{i\}},s_i^{\{i,j\}})=(s_{i}^{\alpha})_{\alpha\in C_i}$.

\end{defn}
                   
\noindent If we assume efforts not to vary across coalitions, then $s_i^{\{i,j\}}(t_l)=s_i^{\{i\}}(t_l)$ and $s_i^{\{i,j\}}(t_h)=s_i^{\{i\}}(t_h)$, $i=1,2$. We have analyzed this situation in the earlier sections. We explain the payoff from a strategy profile below.
Fix a strategy profile $((s_{1}^{\alpha})_{\alpha\in C_1},(s_{2}^{\alpha})_{\alpha\in C_2})$ and a type profile $(t_1,t_2)$. The worth of admissible coalitions corresponding to the  corresponding game are: 
$M((s_{1}^{\{1,2\}}(t_1), s_{2}^{\{1,2\}}(t_2)), (t_1,t_2))$, $M((s_{1}^{\{1\}}(t_1),0), (t_1,0)), 
M((0,s_{2}^{\{2\}}(t_2)), (0,t_2,)),M((0,0), (0,0))$. The Shapley value of employee $1$ corresponding to 
this game is denoted by, $$Sh_{1}((s_{1}^{\alpha})_{\alpha\in C_1},(s_{2}^{\alpha})_{\alpha\in C_2},t_{1},t_{2}))$$

The expected payoff of employee $1$ for the play of the strategy-profile 
$((s_{1}^{\alpha})_{\alpha\in C_1},(s_{2}^{\alpha})_{\alpha\in C_2})$
is:$\Pi_{1}((s_{1}^{\alpha})_{\alpha\in C_1},(s_{2}^{\alpha})_{\alpha\in C_2})=$ \[ \sum_{ (t_{1},t_{2})\in T^{2}   }  [Sh_{1}((s_{1}^{\alpha})_{\alpha\in C_1},(s_{2}^{\alpha})_{\alpha\in C_2},t_{1},t_{2}))-C(s_{1}^{\{1\}}(t_{1}),t_{1})-C(s_{1}^{\{1,2\}}(t_{1}),t_{1})]p(t_{1})p(t_{2}).\] 

\noindent The expected payoff of employee $2$ is computed analogously. 
The notion of SNE is defined analogously. We note that a deviation by a player 
can occur in many ways. For instance $(s_i^{\{i\}'},s_{i}^{\{i,j\}})$ is a deviation from $(s_i^{\{i\}},s_{i}^{\{i,j\}})$, if $s_i^{\{i\}'}$ is a function that is distinct from $s_i^{\{i\}}$. 
Since the next proposition provides a characterization of symmetric equilibria, we drop the suffix $i$ from the strategies. For $\alpha\in C_i$, the function $s^{\alpha}_{e_1e_2}$ is 
defined as $s^{\alpha}_{e_1e_2}(t_l)=e_1,
s^{\alpha}_{e_1e_2}(t_h)=e_2 $.
The next proposition says that it is possible that in equilibrium a player may put different efforts across coalitions. 

\begin{prop}\rm Let $M$ be super-modular and $C$ be sub-modular. Further let $M$ be concave within type. Then  for any probability distribution over $\{t_l,t_h\}$ exactly one of the following holds. 
\begin{enumerate}
		\item [(i)]  If there are SNEs in the strategy      $(s_{e_{l}e_{l}}^{\{1\}},s_{e_{l}e_{l}}^{\{1,2\}})$  or 
  $(s_{e_{l}e_{h}}^{\{1\}},s_{e_{l}e_{l}}^{\{1,2\}})$, then there are no other SNEs.
		\item [(ii)] If there are SNEs in the strategy    $(s_{e_{h}e_{h}}^{\{1\}},s_{e_{l}e_{l}}^{\{1,2\}})$ or    $(s_{e_{l}e_{h}}^{\{1\}},s_{e_{l}e_{l}}^{\{1,2\}})$,  
		then there are no other SNEs.
		\item  [(ii)] If there are SNEs in the strategy $(s_{e_{l}e_{l}}^{\{1\}},s_{e_{l}e_{h}}^{\{1,2\}})$ or  $(s_{e_{l}e_{h}}^{\{1\}},s_{e_{l}e_{h}}^{\{1,2\}})$, then 
		there are no other SNEs.
		\item [(iv)] If there are SNEs in the strategy $(s_{e_{h}e_{h}}^{\{1\}},s_{e_{l}e_{h}}^{\{1,2\}})$   or   $(s_{e_{l}e_{h}}^{\{1\}},s_{e_{l}e_{h}}^{\{1,2\}})$, then 
		there are no other SNEs.  
		\item [(v)] If there are SNEs in the strategy $(s_{e_{h}e_{h}}^{\{1\}},s_{e_{h}e_{h}}^{\{1,2\}})$ or   $(s_{e_{l}e_{h}}^{\{1\}},s_{e_{h}e_{h}}^{\{1,2\}})$, then  
		there are no other SNEs.    
		\item  [(vi)]If there are SNEs in the strategy $(s_{e_{l}e_{l}}^{\{1\}},s_{e_{h}e_{h}}^{\{1,2\}})$  or  $(s_{e_{l}e_{h}}^{\{1\}},s_{e_{h}e_{h}}^{\{1,2\}})$, then 
		there are no other SNEs.    
\end{enumerate}
\label{prop:strategic_coaltions}
\end{prop}

\noindent From (ii) in Proposition \ref{prop:strategic_coaltions} we can see that an employee puts in high effort when she works alone in the machine, and she puts in low effort when she works jointly with the other employee. 
Only one of the cases in Proposition \ref{prop:strategic_coaltions} can arise. In fact, under some mild conditions this is shown Section $5$ in \url{https://drive.google.com/file/d/1M5P3AfTVUpJvhQc2Pj8FL6GgXnUkz51b/view?usp=sharing}. For each of the cases, there is an equilibrium where the effort of an employee differs across coalitions. However, in each equilibria in each case the efforts of an employee in the grand coalitions are the same. If the output of the grand coalition is what matters for the firm management, then 
the differences in the efforts across coalitions only affect the marginal contributions, and thus the wages of the employees.

\section{\bf Concluding Remarks} 
\label{sec:con}

We consider a firm whose production technology is intelligent. The management of the firm hires two
employees. The management of the firm does not know how efficient an employee is when hiring her.
This entails a game of incomplete information, where corresponding to each type profile there is a
stage game in which the available actions for the employees are their effort levels. The payoff function
of the employees incorporates a cooperative, and a non-cooperative aspect. Our characterization
results explain the dependence of symmetric ex-ante equilibria on the distribution of type profiles.
Hence, our result can be considered to be a formal exposition of the relationship between the distribution of efficiency levels and the distribution of outputs. 
Since in our model both efficiency levels and effort levels are not observable before the output is observed, our model incorporates both hidden information and moral hazard.
Nonetheless, our solution to this incomplete information problem is simple. The simplicity arises because the production technology in our model is intelligent.
This approach is different from the models that depend on the asymmetric information approach, i.e., employees know their efficiency levels but the management does not know. The assumption of asymmetric information leads to the consideration of incentive compatibility as a solution concept for the  incomplete information problem. Our approach, which is based on the notion of an intelligent machine provides us with a framework that enables us to not assume that employees know their efficiency levels. 
 We note that asymmetric information need not be a strong assumption in every context. For example, in a model where the preference of an individual is not common knowledge,  asymmetric information need not be a strong assumption to have.
  But how efficient an individual truly is something that need not be known to the individual. In such a context, our notion of an intelligent machine has a role to play since incentive-compatible mechanisms in the context of hiring require employees to report who they are truthfully. 
  In a nutshell, our paper is an attempt to analyze the direct implications of artificial intelligence for equilibria outcomes, and in particular for the decision problem that firms face regarding whom to hire.

\newpage

\section*{\bf Appendix}

\noindent{\bf Proof of Proposition \ref{prop:SNE_ineq_hh}}

\noindent An outline of the steps in the proof are follows. First we show that $(s_{e_he_h},s_{e_he_h})$ is an SNE of $\Gamma_{p}$ for some $p \in \mathbb{P}$ if and only if $$\Delta C_{t_l} \leq p(t_l) \Delta_{e_le_h}^{e_he_h} Sh_1(t_lt_l)+p(t_h)  \Delta_{e_le_h}^{e_he_h} Sh_1(t_lt_h).$$ Then we show that there is a probability distribution $p\in \mathbb{P}$ for which $$\Delta C_{t_l} \leq p(t_l) \Delta_{e_le_h}^{e_he_h} Sh_1(t_lt_l)+p(t_h)  \Delta_{e_le_h}^{e_he_h} Sh_1(t_lt_h)$$  holds if $$\Delta_{e_le_h}^{e_he_h} Sh_1(t_lt_l) < \Delta C_{t_l} <  \Delta_{e_le_h}^{e_he_h} Sh_1(t_lt_h)$$ holds. This establishes existence of an SNE for some $p\in\mathbb{P}$. Then we also show that if $\Delta_{e_le_h}^{e_he_h} Sh_1(t_lt_l) < \Delta C_{t_l} <  \Delta_{e_le_h}^{e_he_h} Sh_1(t_lt_h)$ holds then there is $p\in \mathbb{P}$ for which $\Delta C_{t_l} \leq p(t_l) \Delta_{e_le_h}^{e_he_h} Sh_1(t_lt_l)+p(t_h)  \Delta_{e_le_h}^{e_he_h} Sh_1(t_lt_h)$ does not hold; we use the Farkas' lemma to show this. This establishes that sufficiency of $\Delta_{e_le_h}^{e_he_h} Sh_1(t_lt_l) < \Delta C_{t_l} <  \Delta_{e_le_h}^{e_he_h} Sh_1(t_lt_h)$ in Proposition \ref{prop:SNE_ineq_hh}. Next we assume that $(i)$ and $(ii)$ hold and use Farkas' lemma to show that  $\Delta_{e_le_h}^{e_he_h} Sh_1(t_lt_l) < \Delta C_{t_l} <  \Delta_{e_le_h}^{e_he_h} Sh_1(t_lt_h)$ holds.


\begin{lemma}\rm Let $M$ be super-modular $C$ be sub-modular and $p \in \mathbb{P}$, $(s_{e_he_h},s_{e_he_h})$ is an SNE of $\Gamma_{p}$ if and only if $ \Delta C_{t_l} \leq p(t_l) \Delta_{e_le_h}^{e_he_h} Sh_1(t_lt_l)+p(t_h)  \Delta_{e_le_h}^{e_he_h} Sh_1(t_lt_h)$.
	\label{lemma:nec_suff_hh}
\end{lemma}
\noindent{\bf Proof of Lemma \ref{lemma:nec_suff_hh}} Let $(s_{e_he_h},s_{e_he_h})$ be an SNE for the game $\Gamma_{p}$, then 
\begin{equation}\label{eqHH1}
\Pi_1((s_{e_he_h},s_{e_he_h})) \geq \Pi_1(  s_{e_le_h},s_{e_he_h})
\end{equation}
$\Leftrightarrow p(t_l)p(t_l) \{Sh_1((e_h,e_h),(t_l,t_l))-C(e_h,t_l)\}+ p(t_l)p(t_h) \{Sh_1((e_h,e_h),(t_l,t_h))-C(e_h,t_l)\}$
\[\geq   p(t_l)p(t_l) \{Sh_1((e_l,e_h),(t_l,t_l))-C(e_l,t_l)\}+ p(t_l)p(t_h) \{Sh_1((e_l,e_h),(t_l,t_h))-C(e_l,t_l)\}\]
$p(t_l) [p(t_l) \{Sh_1((e_h,e_h),(t_l,t_l))-C(e_h,t_l)-Sh_1((e_l,e_h),(t_l,t_l)) +C(e_l,t_l)\}$
\begin{equation}
+p(t_h) \{Sh_1((e_h,e_h),(t_l,t_h))-Sh_1((e_l,e_h),(t_l,t_h))-C(e_h,t_l)+C(e_l,t_l)\}]\geq 
\label{eqhh2}
\end{equation}
\[\Leftrightarrow   p(t_l) \Delta_{e_le_h}^{e_he_h} Sh_1(t_lt_l)+p(t_h) \Delta_{e_le_h}^{e_he_h} Sh_1(t_lt_h)  \geq \Delta C_{t_l}.\] 

\noindent Therefore player $1$ does not have an incentive to deviate to $s_{e_{l}e_{h}}$ if and only if $\Delta C_{t_l} \leq p(t_l)\Delta_{e_le_h}^{e_he_h} Sh_1(t_lt_l)+ p(t_h) \Delta_{e_le_h}^{e_he_h} Sh_1(t_lt_h)$, given this we show next that other deviations are also not profitable for player $1$. Since $C$ is sub-modular,  $\Delta C_{t_l}=C(e_h,t_l)-C(e_l,t_l)> \Delta C_{t_h}= C(e_h,t_h)-C(e_l,t_h)$.  Since $M$ is super modular, $\Delta_{e_le_h}^{e_he_h} Sh_1(t_ht_l) > \Delta_{e_le_h}^{e_he_h} Sh_1(t_lt_l) $ and $\Delta_{e_le_h}^{e_he_h} Sh_1(t_ht_h) > \Delta_{e_le_h}^{e_he_h} Sh_1(t_lt_h)$. Now $\Delta C_{t_h} < p(t_l) \Delta_{e_le_h}^{e_he_h} Sh_1(t_lt_l)+  p(t_h) \Delta_{e_le_h}^{e_he_h} Sh_1(t_lt_h)$. 
\[\Rightarrow p(t_l)\{\Delta C_{t_h} - \Delta_{e_le_h}^{e_he_h} Sh_1(t_ht_l)\}+  p(t_h)\{\Delta C_{t_h} - \Delta_{e_le_h}^{e_he_h} Sh_1(t_ht_h)\} < 0\]
\begin{equation*}	
\Leftrightarrow p(t_h) [p(t_l) \{Sh_1((e_l,e_h),(t_h,t_l))-C(e_l,t_h)-Sh_1((e_h,e_h),(t_h,t_l)) +C(e_h,t_h)\}\end{equation*}
\begin{equation}
+ p(t_h) \{Sh_1((e_l,e_h),(t_h,t_h))-Sh_1((e_h,e_h),(t_h,t_h))-C(e_l,t_h)+C(e_h,t_h)\}] < 0
\label{eqhh3}
\end{equation}	
\[\Leftrightarrow \Pi_1(s_{e_he_l},s_{e_he_h}) < \Pi_1(s_{e_he_h},s_{e_he_h})\]

\noindent Therefore player $1$ does not have a profitable deviation opportunity to $s_{e_{l}e_{h}}$ if and only if $\Delta C_{t_l} \leq p(t_l)\Delta_{e_le_h}^{e_he_h} Sh_1(t_lt_l)+ p(t_h) \Delta_{e_le_h}^{e_he_h} Sh_1(t_lt_h)$ implies player $1$ does not have a profitable  deviation opportunity to $s_{e_{h}e_{l}}$. From $\Delta C_{t_l} \leq p(t_l) \Delta_{e_le_h}^{e_he_h} Sh_1(t_lt_l)+p(t_h) \Delta_{e_le_h}^{e_he_h} Sh_1(t_lt_h)$,  sub-modular  $C$ and  Super-modular $M$ it follows that Inequality (\ref{eqhh2}) + inequality (\ref{eqhh3})>0 and hence player $1$ does not have any incentive to unilaterally deviate to the strategy $s_{e_{l}e_{l}}$. \\
\noindent{\bf End of Proof of Lemma \ref{lemma:nec_suff_hh}} 
\medskip 

\begin{lemma}\rm Let $M$ be super-modular and $C$ be sub-modular. If $\Delta_{e_le_h}^{e_he_h} Sh_1(t_lt_l) < \Delta C_{t_l} <  \Delta_{e_le_h}^{e_he_h} Sh_1(t_lt_h)$ then there is $p \in \mathbb{P}$ such that $ \Delta C_{t_l} \leq p(t_l) \Delta_{e_le_h}^{e_he_h} Sh_1(t_lt_l)+p(t_h)  \Delta_{e_le_h}^{e_he_h} Sh_1(t_lt_h)$ .
	\label{lemma:Existence-inequality_hh}
\end{lemma}
\medskip
\noindent{\bf Proof of Lemma \ref{lemma:Existence-inequality_hh}} It is enough to show that the following system  has a solution.   
\begin{eqnarray*}
	&& p(t_l)\{\Delta C_{t_l}- \Delta_{e_le_h}^{e_he_h} Sh_1(t_lt_l)\}+  p(t_h)\{\Delta C_{t_l}-\Delta_{e_le_h}^{e_he_h} Sh_1(t_lt_h)\} \leq 0\\
	&& p(t_l)> 0,~ p(t_h)> 0,~ p(t_l) + p(t_h) = 1
\end{eqnarray*} 
\noindent This system of inequalities is rewritten as, call the rewritten system $(Ph)$: 
\begin{equation*}
p(t_l)\{\Delta C_{t_l}- \Delta_{e_le_h}^{e_he_h} Sh_1(t_lt_l)\}+  p(t_h)\{\Delta C_{t_l}-\Delta_{e_le_h}^{e_he_h} Sh_1(t_lt_h)\} \leq  0
\end{equation*}
\begin{equation}
\tag{Ph}
p(t_l) +p(t_h) \leq 1
\end{equation}
\begin{equation*}
-p(t_l) -p(t_h) \leq -1
\end{equation*}
\begin{equation*}
p(t_l).0-p(t_h)< 0
\end{equation*}
\begin{equation*}
-p(t_l) + p(t_h).0  < 0.
\end{equation*}

\noindent This system then can be seen succinctly in the form $Ax \leq b$, $Bx< c$,  with $x \equiv (p(t_l),p(t_h))$ with;

$\tiny{A= \begin{bmatrix}
	\Delta C_{t_l}- \Delta_{e_le_h}^{e_he_h} Sh_1(t_lt_l) &  \Delta C_{t_l}-\Delta_{e_le_h}^{e_he_h} Sh_1(t_lt_h)\\
	1 & 1 \\
	-1 & -1  
	\end{bmatrix},
	b= \begin{bmatrix}
	0\\
	1\\
	-1
	\end{bmatrix}, B= \begin{bmatrix}
	0 & -1\\
	-1 & 0\\
	\end{bmatrix} , c= \begin{bmatrix}
	0\\
	0\\
	\end{bmatrix}}$\\

\noindent In order to establish that this system has solution we show that the appropriate Farkas' dual does not have a solution. We recall the  following version of Farkas' lemma.  

\begin{theorem}\rm[{\bf Farkas' Lemma}~ \citep{Motzkin}] Exactly one of the following statements is true.
	\begin{enumerate}
		\item [$(1)$] There exists  $x$ satisfying $Ax\leq b$ and $Bx <c$.
		\item [$(2)$] There exist $y,z$ such that,
		$$y \geq 0,~~ z\geq 0,~~A^{T}y+B^{T}z=0$$ and $$ b^{T}y+c^{T}z <0~\text{or}~b^{T}y+c^{T}z =0,~~ z \neq 0$$
	\end{enumerate}
	\label{thm:Farkas}
\end{theorem}

\noindent The corresponding Farkas' dual, i.e. the set of inequalities $(2)$ in Theorem \ref{thm:Farkas}, for our system of inequalities is the following:    

$y=(y_1,y_2,y_3), y_{i}\geq 0, i=1,2,3,~z = (z_1,z_2), z_{j}\geq 0, j=1,2$;
\[A^{T}y+B^{T}z=0~\text{and}~ b^{T}y+c^{T}z <0~\text{or}~b^{T}y+c^{T}z =0~;~z\neq 0.\]

\noindent We show this dual has no solution; and hence a proof of Lemma \ref{lemma:Existence-inequality_hh} follows by Theorem \ref{thm:Farkas}. We prove this claim by the way of contradiction.

\noindent We note $ b^{T}y+c^{T}z <0 \Leftrightarrow y_2 < y_3$.
\begin{eqnarray*}
	& & A^{T}y+B^{T}z=0 \Leftrightarrow  \begin{bmatrix}
		\Delta C_{t_l}- \Delta_{e_le_h}^{e_he_h} Sh_1(t_lt_l)  & 1 & -1 \\
		\Delta C_{t_l}-\Delta_{e_le_h}^{e_he_h} Sh_1(t_lt_h)&  1 & -1 
	\end{bmatrix}.\begin{bmatrix}
		y_1\\
		y_2\\
		y_3
	\end{bmatrix}+\begin{bmatrix}
		0 & -1 \\
		-1 & 0
	\end{bmatrix}.\begin{bmatrix}
		z_1\\
		z_2
	\end{bmatrix}=0\\
	& \Leftrightarrow & \{\Delta C_{t_l}- \Delta_{e_le_h}^{e_he_h} Sh_1(t_lt_l)  \}y_1 +  y_2 -y_3 -z_2=0\\
	& & \{ \Delta C_{t_l}-\Delta_{e_le_h}^{e_he_h} Sh_1(t_lt_h) \} y_1  +y_2-y_3 -z_1=0\\
	& \Leftrightarrow &  \{\Delta C_{t_l}-\Delta_{e_le_h}^{e_he_h} Sh_1(t_lt_h) \}y_1   -z_1= y_3 - y_2> 0 ~\text{ contradiction as } ~ \Delta C_{t_l} <  \Delta_{e_le_h}^{e_he_h} Sh_1(t_lt_h). 
\end{eqnarray*}
\noindent Now consider the case,   
\begin{eqnarray*}
	&& A^{T}y+B^{T}z=0~;~b^{T}y+c^{T}z =0\Leftrightarrow y_2=y_3,~ z \neq 0 \Leftrightarrow z_1 >  0 ~\text{or}~ z_2 > 0.\\
	& & \text{Then,}~~ \{\Delta C_{t_l}- \Delta_{e_le_h}^{e_he_h} Sh_1(t_lt_l)   \}y_1  + y_2 -y_3 -z_2=0\\
	& & \{ \Delta C_{t_l}-\Delta_{e_le_h}^{e_he_h} Sh_1(t_lt_h) \} y_1  +y_2-y_3 -z_1=0\\
	& \Leftrightarrow & \{\Delta C_{t_l}- \Delta_{e_le_h}^{e_he_h} Sh_1(t_lt_l)  \}y_1   -z_2=0\\
	& & \{\Delta C_{t_l}-\Delta_{e_le_h}^{e_he_h} Sh_1(t_lt_h)  \} y_1  -z_1=0\\
	& \Leftrightarrow &  \{\Delta C_{t_l}-\Delta_{e_le_h}^{e_he_h} Sh_1(t_lt_h)   \}y_1 -z_1= 0 ~\text{ contradiction if } ~ y_{1}\neq 0~\text{as}~ \Delta C_{t_l} < \Delta_{e_le_h}^{e_he_h} Sh_1(t_lt_h) ;\\ 
	&&\text{if}~y_1=0~\text{then}~z_1=z_2=0.  
\end{eqnarray*}
\noindent{\bf End of Proof of Lemma \ref{lemma:Existence-inequality_hh}}

\medskip 

\noindent{\bf Now we go back to the proof of Proposition \ref{prop:SNE_ineq_hh}.}

\noindent Let $\Delta_{e_le_h}^{e_he_h} Sh_1(t_lt_l) < \Delta C_{t_l} <  \Delta_{e_le_h}^{e_he_h} Sh_1(t_lt_h)$, and we show $(i)$ and $(ii)$ to hold.

\medskip
\noindent We first show $(i)$ in Proposition \ref{prop:SNE_ineq_hh}. Choose the the probability distribution for which Lemma \ref{lemma:Existence-inequality_hh} holds, and then by Lemma \ref{lemma:nec_suff_hh}, $(i)$ follows.

\medskip

\noindent Now we show $(ii)$ in Proposition \ref{prop:SNE_ineq_hh} i.e., we show if 
$\Delta_{e_le_h}^{e_he_h}Sh_1(t_lt_l)<\Delta C_{t_l}<\Delta_{e_le_h}^{e_he_h}Sh_1(t_lt_h)$ then $(ii)$ holds. In particular by Lemma \ref{lemma:nec_suff_hh} it is sufficient to show that there is $p \in \mathbb{P}$ such that $\Delta C_{t_l}\leq p(t_l)\Delta_{e_le_h}^{e_he_h}Sh_1(t_lt_l)+p(t_h)\Delta_{e_le_h}^{e_he_h}Sh_1(t_lt_h)$ does not hold i.e., $p(t_l) \{ \Delta C_{t_l} -\Delta_{e_le_h}^{e_he_h}Sh_1(t_lt_l)\}+p(t_h) \{ \Delta C_{t_l} -\Delta_{e_le_h}^{e_he_h}Sh_1(t_lt_h)\} > 0$. In order to show this we show that the following system of inequalities has a solution.  
\begin{equation*}
p(t_l) \{ \Delta_{e_le_h}^{e_he_h}Sh_1(t_lt_l) -\Delta C_{t_l}\}+p(t_h) \{ \Delta_{e_le_h}^{e_he_h}Sh_1(t_lt_h) -\Delta C_{t_l} \} < 0
\end{equation*}
\begin{equation*}
-p(t_l) + 0.p(t_h) < 0
\end{equation*}
\begin{equation}
\tag{Qh}
0.p(t_l)-p(t_h) < 0
\end{equation}
\begin{equation*}
p(t_l) + p(t_h) \leq 1
\end{equation*}
\begin{equation*}
-p(t_l) -p(t_h) \leq -1
\end{equation*}
\medskip

\noindent Which is of the form $Ax \leq b$, $Bx< c$, 
with $x  \equiv (p(t_l),p(t_h))$, \\
$\tiny {A= \begin{bmatrix}
	1 & 1 \\
	-1 & -1  
	\end{bmatrix},
	b= \begin{bmatrix}
	1\\
	-1
	\end{bmatrix}, B= \begin{bmatrix}
	\Delta_{e_le_h}^{e_he_h}Sh_1(t_lt_l) -\Delta C_{t_l}& \Delta_{e_le_h}^{e_he_h}Sh_1(t_lt_h) -\Delta C_{t_l}\\
	-1 &0\\
	0& -1
	\end{bmatrix} , c= \begin{bmatrix}
	0\\
	0\\
	0\\
	\end{bmatrix}}$
\medskip

\noindent The Dual of $(Qh)$ is\\
$y \equiv (y_1,y_2) \geq 0,~z \equiv (z_1,z_2,z_3)\geq 0;~A^{T}y+B^{T}z=0~;~ b^{T}y+c^{T}z <0~\text{or}~b^{T}y+c^{T}z =0~;~z\neq 0$.

\noindent We show that this dual has no solution. We prove this claim by the way of contradiction. First consider $b^{T}y+c^{T}z <0$.
\begin{eqnarray*}
	& & A^{T}y+B^{T}z=0;~b^{T}y+c^{T}z <0  \Leftrightarrow y_1 < y_2 \\
	& & \begin{bmatrix}
		1 & -1\\
		1 & -1
	\end{bmatrix}.\begin{bmatrix}
		y_1\\				
		y_2			
	\end{bmatrix} + \begin{bmatrix}
		\Delta_{e_le_h}^{e_he_h}Sh_1(t_lt_l) -\Delta C_{t_l} & -1 & 0\\
		\Delta_{e_le_h}^{e_he_h}Sh_1(t_lt_h) -\Delta C_{t_l} &  0 &-1
	\end{bmatrix}.\begin{bmatrix}
		z_1\\
		z_2\\
		z_3
	\end{bmatrix}=0\\
	& \Leftrightarrow & y_1 -y_2 + \{\Delta_{e_le_h}^{e_he_h}Sh_1(t_lt_l) -\Delta C_{t_l}\}.z_1-z_2=0\\
	& & y_1 -y_2 + \{\Delta_{e_le_h}^{e_he_h}Sh_1(t_lt_h) -\Delta C_{t_l}\}.z_1-z_3=0,~\text{a contradiction as}~
	\Delta_{e_le_h}^{e_he_h}Sh_1(t_lt_l) < \Delta C_{t_l}.
\end{eqnarray*}
\noindent Now consider the situation in which $b^{T}y+c^{T}z=0$. That is,
\begin{eqnarray*}
	& & b^{T}y+c^{T}z =0 \Leftrightarrow y_1=y_2,~ z \neq 0 \Leftrightarrow z_1> 0 ~\text{or}~z_2> 0 ~\text{or}~ z_3 > 0\\
	& \Leftrightarrow & y_1 -y_2 + \{\Delta_{e_le_h}^{e_he_h}Sh_1(t_lt_l) -\Delta C_{t_l} \}.z_1-z_2=0\\
	& & y_1 -y_2 + \{\Delta_{e_le_h}^{e_he_h}Sh_1(t_lt_h) -\Delta C_{t_l}\}.z_1-z_3=0\\
	& \Leftrightarrow &  \{\Delta_{e_le_h}^{e_he_h}Sh_1(t_lt_l) -\Delta C_{t_l}\}.z_1-z_2=0\\
	& & \{\Delta_{e_le_h}^{e_he_h}Sh_1(t_lt_h) -\Delta C_{t_l}\}.z_1-z_3=0~
	\text{ contradiction if } ~ z_{1}\neq 0~\text{as}~ \Delta_{e_le_h}^{e_he_h}Sh_1(t_lt_l) < \Delta C_{t_l};\\ 
	&&\text{if}~z_1=0~\text{then}~z_3=0=z_{2}.
\end{eqnarray*}

\noindent Therefore the dual of $(Qh)$ does not have a solution, hence $(ii)$ in Proposition \ref{prop:SNE_ineq_hh} is established by Theorem\ref{thm:Farkas}.  

\medskip

\noindent {\bf Now we prove the converse of Proposition \ref{prop:SNE_ineq_hh}.} That is we assume $(i)$ and $(ii)$ to hold; and show $\Delta_{e_le_h}^{e_he_h} Sh_1(t_lt_l) < \Delta C_{t_l} <\Delta_{e_le_h}^{e_he_h} Sh_1(t_lt_h)$. Since $(i)$ and $(ii)$ hold the duals of $Ph$ and $Qh$ do not have a solution by Theorem \ref{thm:Farkas}. First we note that in both systems of inequalities $Ph$ and $Qh$ the scalar vector $b$ has both positive and negative entries. This means that the reason for the duals of $Ph$ and $Qh$ not to have solutions is not because $b^{T}y+c^{T}z>0$ for all $y\geq 0, z\geq 0$ and $z\neq 0$. 
Also it is not true that $b^{T}y+c^{T}z=0$ implies $z=0$. Hence it is enough to  consider the situations described by the duals of $Ph$ and $Qh$ and look at the implications if they are violated in order to establish the converse of Proposition \ref{prop:SNE_ineq_hh}. First the following intermediary result is needed. 
Let $A=\Delta_{e_le_h}^{e_he_h} Sh_1(t_lt_l)-\Delta C_{t_l}$ and $B=\Delta_{e_le_h}^{e_he_h} Sh_1(t_lt_h)-\Delta C_{t_l}$. 

\begin{lemma}\rm For $(M,C)$ the following holds,
	
	\begin{itemize}
		\item [$(a)$] If $A\geq 0$ then $B>0$.
		\item [$(b)$] If $B\leq 0$ then $A<0$.
	\end{itemize} 
	\label{lemma:shapley_cost_hh}
\end{lemma} 
\noindent{\bf Proof of Lemma \ref{lemma:shapley_cost_hh}} Since $M$ is super-modular  therefore, $\Delta_{e_le_h}^{e_he_h} Sh_1(t_lt_h) >   \Delta_{e_le_h}^{e_he_h} Sh_1(t_lt_l ) $
\noindent {\bf Proof of $(a)$:} Let, \begin{eqnarray*}
	& & A \geq 0 \Rightarrow  \Delta_{e_le_h}^{e_he_h} Sh_1(t_lt_l ) \geq  \Delta C_{t_l}  \\
	&\Rightarrow &  \Delta_{e_le_h}^{e_he_h} Sh_1(t_lt_h) >   \Delta_{e_le_h}^{e_he_h} Sh_1(t_lt_l ) \geq \Delta C_{t_l}   \\
	&\Rightarrow & \Delta_{e_le_h}^{e_he_h} Sh_1(t_lt_h ) -  \Delta C_{t_l} >0  \Leftrightarrow B > 0
\end{eqnarray*}
\noindent {\bf End of the   Proof of $(a)$.}\\
\noindent {\bf Proof of $(b)$:}
If \begin{eqnarray*}
	& & B \leq 0 \Rightarrow  \Delta_{e_le_h}^{e_he_h} Sh_1(t_lt_h ) \leq   \Delta C_{t_l}  \\
	&\Rightarrow &   \Delta_{e_le_h}^{e_he_h} Sh_1(t_lt_l )  < \Delta_{e_le_h}^{e_he_h} Sh_1(t_lt_h) \leq \Delta C_{t_l}   \\
	&\Rightarrow & \Delta_{e_le_h}^{e_he_h} Sh_1(t_lt_l ) -  \Delta C_{t_l} < 0 \Leftrightarrow  A < 0
\end{eqnarray*}
\noindent {\bf End of the Proof of $(b)$\\}

\noindent{\bf End of Proof of Lemma \ref{lemma:shapley_cost_hh}}
\medskip	
\noindent{\bf Now we go back to the proof of the converse in Proposition \ref{prop:SNE_ineq_hh}}. 

\noindent{\bf Consider the dual of $(Ph)$.}

\noindent {\bf Step $1$:} We note $b^{T}y+c^{T}z<0 \Leftrightarrow y_3>y_2$; and the dual does not have a solution means for all $y\geq 0,z\geq 0$  the following system has no solution, which in turn implies at least one of them has no solution.  
\begin{eqnarray*}
	& & \{\Delta_{e_le_h}^{e_he_h}Sh_1(t_lt_l) -\Delta C_{t_l}  \}y_1 + y_2 -y_3 -z_2=0\\
	& & \{ \Delta_{e_le_h}^{e_he_h}Sh_1(t_lt_h) -\Delta C_{t_l}  \} y_1  +y_2- y_3 -z_1=0
\end{eqnarray*} 
\noindent 
We argue that if $A=\Delta_{e_le_h}^{e_he_h} Sh_1(t_lt_l)-\Delta C_{t_l} >0$ and $B= \Delta_{e_le_h}^{e_he_h} Sh_1(t_lt_h)-\Delta C_{t_l}>0$ then the dual has a solution. To see this let $A>0,B>0$ and consider $A=y_{3}-y_{2}+z_{2}$ and $B=y_{3}-y_{2}+z_{1}$. Since $y\geq 0,z\geq 0$ $y_{3}-y_{2}+z_{2}\in (0,\infty)$ and $y_{3}-y_{2}+z_{1}\in (0,\infty)$, and hence choose $A=y_{3}^{*}-y_{2}^{*}+z_{2}^{*}$ and 
$0<y_{3}^{*}-y_{2}^{*}<B$. Then set $z_{1}^{*}=B-[y_{3}^{*}-y_{2}^{*}]$. This means $(1,y_{2}^{*},y_{3}^{*}),(z_{1}^{*},z_{2}^{*})$ is a solution to the dual of $(Ph)$.   

\noindent Hence either $\Delta_{e_le_h}^{e_he_h} Sh_1(t_lt_l) \leq \Delta C_{t_l}$ or $ \Delta_{e_le_h}^{e_he_h} Sh_1(t_lt_h) \leq  \Delta C_{t_l}$.

\medskip

\noindent {\bf Step $2$: }Analogous to the last step, for  $b^{T}y+c^{T}z=0 \Leftrightarrow y_3 = y_2$; and the dual does not have a solution means for all $y\geq 0,z\geq 0,z\neq 0$ at least one equation in the system 
\begin{eqnarray*}
	&& \{A=[\Delta_{e_le_h}^{e_he_h} Sh_1(t_lt_l)-\Delta C_{t_l} ] \}y_1 + y_2 -y_3 - z_2=0\\
	& &\{B= [\Delta_{e_le_l}^{e_he_l} Sh_1(t_lt_h)-\Delta C_{t_l}] \} y_1+y_2-y_3 -z_1=0\\
	& \Leftrightarrow & A y_1- z_2=0 ; B y_1-z_1=0
\end{eqnarray*}	
\noindent does not hold. By Lemma \ref{lemma:shapley_cost_hh} if $A=0$, then $B>0$. Then setting $z_{2}=0$ we can find a solution to the dual. Hence $A\neq 0$. 
If $A>0$, by Step $1$ $B\leq 0$; which contradicts Lemma \ref{lemma:shapley_cost_hh} since $B\leq 0$ implies $A<0$ by Lemma \ref{lemma:shapley_cost_hh}. Also by Lemma \ref{lemma:shapley_cost_hh} if $A>0$ then $B>0$, which contradicts Step $1$.

\noindent Hence the only possibility is $A<0$ and $B\geq 0$ , i.e. $\Delta_{e_le_h}^{e_he_h} Sh_1(t_lt_l)<\Delta C_{t_l}$ and $\Delta C_{t_l}\leq \Delta_{e_le_h}^{e_he_h} Sh_1(t_lt_h)$.

\medskip  

\noindent{\bf Now consider the dual of $(Qh)$. }

\noindent{\bf Step $3$:} We note $b^{T}y+c^{T}z<0 \Leftrightarrow y_2 > y_1$, and the dual of $(Qh)$ does not have a solution means the system of inequalities for all $y\geq 0,z\geq 0$ at least one equation in the system   
\begin{eqnarray*}
	& & y_1 -y_2 + \{\Delta_{e_le_h}^{e_he_h}Sh_1(t_lt_l)-\Delta C_{t_l} \}.z_1-z_2=0\\
	& & y_1 -y_2 + \{\Delta_{e_le_h}^{e_he_h}Sh_1(t_lt_h)-\Delta C_{t_l}\}.z_1-z_3=0
\end{eqnarray*}
\noindent does not hold. By an argument analogous to Step $1$ either  
$\Delta C_{t_l} \geq \Delta_{e_le_h}^{e_he_h}Sh_1(t_lt_h)$ or $\Delta C_{t_l} \geq \Delta_{e_le_h}^{e_he_h}Sh_1(t_lt_l)$.

\medskip

\noindent{\bf Step $4$:} Also $b^{T}y+c^{T}z=0 \Leftrightarrow y_2 = y_1$, and the dual of $(Qh)$ does not have a solution means the system of inequalities for all $y\geq 0,z\geq 0, z\neq 0$, has no solution. Which in turn implies following equation has no solution.
\begin{eqnarray*}
	&& y_1 -y_2 + \{\Delta_{e_le_h}^{e_he_h}Sh_1(t_lt_l)-\Delta C_{t_l} \}.z_1-z_2=0\\
	& & y_1 -y_2 + \{\Delta_{e_le_h}^{e_he_h}Sh_1(t_lt_h)-\Delta C_{t_l} \}.z_1-z_3=0\\
	& \Leftrightarrow &  -Az_1-z_2=0 ~\text{and}~ -Bz_1-z_3=0\\
	& \Leftrightarrow &  Az_1+z_2=0  ~\text{and}~  Bz_1+z_3=0
\end{eqnarray*} 
\noindent By Lemma \ref{lemma:shapley_cost_hh}, if $B<0$ then $A<0$. In this situation we can find a solution to the dual of $(Qh)$. In particular $y=(y_{1},y_{2}), y_{1}=y_{2}$, $z=(1, -A, -B )$ is a solution. Hence $B\geq 0$. If $B=0$, then by Lemma \ref{lemma:shapley_cost_hh}, $A<0$. Then $y=(y_{1},y_{2}), y_{1}=y_{2}$, $z=(1, -A, 0 )$ is a solution to the dual of $(Qh)$. If $B>0$ and $A>0$ then contradicts Step $1$. Hence, $B>0$ and $A \leq 0$. 

\medskip

\noindent Now from Step $2$ and Step $4$ its follows $A<0,B>0$. 
Hence we have established that if $(i)$ and $(ii)$ hold then 
$\Delta_{e_le_h}^{e_he_h}Sh_1(t_lt_l)<\Delta C_{t_l}<\Delta_{e_le_h}^{e_he_h}Sh_1(t_lt_h)$.\\        

\medskip

\noindent {\bf{End of the proof of Proposition \ref{prop:SNE_ineq_hh}}}.

\medskip

\noindent\textbf{Proof of Corollary \ref{Interval-hh}: } From Lemma \ref{lemma:nec_suff_hh} it follows that $(s_{e_he_h},s_{e_he_h})$ is an SNE of $\Gamma_{p}$ if and only if 
\begin{eqnarray*}
	&& p(t_l) \Delta_{e_le_h}^{e_he_h} Sh_1(t_lt_l)+p(t_h)  \Delta_{e_le_h}^{e_he_h} Sh_1(t_lt_h) \geq  \Delta C_{t_l} \\
	&\Leftrightarrow & (1-p(t_h))\Delta_{e_le_h}^{e_he_h} Sh_1(t_lt_l)+p(t_h)  \Delta_{e_le_h}^{e_he_h} Sh_1(t_lt_h) \geq  \Delta C_{t_l} \\
	&\Leftrightarrow & p(t_h) \{ \Delta_{e_le_h}^{e_he_h} Sh_1(t_lt_h)-\Delta_{e_le_h}^{e_he_h} Sh_1(t_lt_l)\} \geq  \{ \Delta C_{t_l} -\Delta_{e_le_h}^{e_he_h} Sh_1(t_lt_l) \} \\
	&\Leftrightarrow & p(t_h) \geq \frac{ \Delta C_{t_l} -\Delta_{e_le_h}^{e_he_h} Sh_1(t_lt_l)}{ \Delta_{e_le_h}^{e_he_h} Sh_1(t_lt_h)-\Delta_{e_le_h}^{e_he_h} Sh_1(t_lt_l)} 
\end{eqnarray*}
Using Proposition \ref{prop:SNE_ineq_hh}, $\Delta_{e_le_h}^{e_he_h} Sh_1(t_lt_l) < \Delta C_{t_l} <  \Delta_{e_le_h}^{e_he_h} Sh_1(t_lt_h) \implies 0 < \frac{ \Delta C_{t_l} -\Delta_{e_le_h}^{e_he_h} Sh_1(t_lt_l)}{ \Delta_{e_le_h}^{e_he_h} Sh_1(t_lt_h)-\Delta_{e_le_h}^{e_he_h} Sh_1(t_lt_l)} < 1$ and hence $(s_{e_he_h},s_{e_he_h})$ is an SNE of $\Gamma_{p}$ if and only if $p(t_h) \in [\frac{\Delta C_{t_l}-\Delta_{e_le_h}^{e_he_h} Sh_1(t_lt_l)}{\Delta_{e_le_h}^{e_he_h} Sh_1(t_lt_h)-\Delta_{e_le_h}^{e_he_h} Sh_1(t_lt_l)},1)$.

\noindent \textbf{End  of the  proof of Corollary \ref{Interval-hh}}  

\noindent\textbf{Tables related to Example \ref{ex:paradox}}

\begin{table}[h!]
	\centering
	\begin{tabular}{|c|c|c|c|c|c|c|}
		\hline 
		$M$  &  $(e_{l},0) $ & $(e_{h},0)$ & $(e_{l},e_{l}) $ & $(e_{l},e_{h}) $ & $(e_{h},e_{l})$ & $(e_{h},e_{h})$  \\
		\hline 
		$(t_{l},0)$ &  7 &  12.9 & & & &  \\
		\hline 
		$(t_{h},0)$ &10   &  16 & & & &  \\
		\hline 
		$(t_{l},t_{l})$ & & & 8 & 11 & 11 & 16 \\
		\hline 
		$(t_{l},t_{h})$ & & & 9  & 13.1 & 13  & 20  \\
		\hline 
		$(t_{h},t_{l})$ & & &  9 & 13 & 13.1 & 20  \\
		\hline 
		$(t_{h},t_{h})$ & & & 10 & 16 & 16 & 25  \\
		\hline 
	\end{tabular}
	\caption{Table for TAM}
\end{table}
\begin{table}[h!]
	\centering
	\begin{tabular}{|c|c|c|}
		\hline 
		$C$  &  $(e_{l},0) $ & $(e_{h},0)$  \\
		\hline 
		$(t_{l},0)$ &  2 &  8   \\
		\hline 
		$(t_{h},0)$ & 1& 6.3   \\
		\hline
	\end{tabular}
	\caption{Table for Cost Function}
\end{table}
\begin{table}[h!]
	\begin{center}\begin{tabular}{|c|c|c|}
			\hline
			Equilibrium & Range of $p_h$ & Total Expected Revenue\\
			\hline
			$(s_{e_le_l},s_{e_le_l})$ & $(0,0.263158]$ & $EW(s_{e_le_l},p)=  0.000000 p_h^2 +2.000000 p_h+2.000000$\\
			\hline
			$(s_{e_le_h},s_{e_le_h})$ & $[0.102041,0.794872]$ & $EW(s_{e_le_h},p)=  3.400000 p_h^2 +0.800000 p_h+2.000000$\\
			\hline
			$(s_{e_he_h},s_{e_he_h})$ & $[0.578948,1)$ & $EW(s_{e_he_h},p)=0.500000 p_h^2 +5.700000 p_h+0.000000$ \\
			\hline
		\end{tabular}
		\caption{Table representing equilibrium range and expected revenue function.}
\end{center}\end{table}
\newpage

\noindent\textbf{Proof of corollary \eqref{rationalizable}}\\
From Corollary \eqref{Interval-hh}, Corollary \eqref{Interval-lh} and Corollary \eqref{Interval-ll} we have, \\
$[\dfrac{\Delta C_{t_l}-\Delta_{e_le_h}^{e_he_h} Sh_1(t_lt_l)}{\Delta_{e_le_h}^{e_he_h} Sh_1(t_lt_h)-\Delta_{e_le_h}^{e_he_h} Sh_1(t_lt_l)},1) \subseteq (0,1),
(0,\dfrac{\Delta C_{t_h}-\Delta_{e_le_l}^{e_he_l} Sh_1(t_ht_l)}{\Delta_{e_le_l}^{e_he_l} Sh_1(t_ht_h)-\Delta_{e_le_l}^{e_he_l} Sh_1(t_ht_l)}] \subseteq (0,1)$,\\ $
[\dfrac{ C_{t_h} -\Delta_{e_le_l}^{e_he_l} Sh_1(t_ht_l)}{\Delta_{e_le_h}^{e_he_h} Sh_1(t_ht_h)-\Delta_{e_le_l}^{e_he_l} Sh_1(t_ht_l)} ,
\dfrac{ C_{t_l} -\Delta_{e_le_l}^{e_he_l} Sh_1(t_lt_l)}{\Delta_{e_le_h}^{e_he_h} Sh_1(t_lt_h)-\Delta_{e_le_l}^{e_he_l} Sh_1(t_lt_l)}] \subseteq (0,1)$. Given $M$ satisfies concavity within each type profile, the following holds.
\begin{eqnarray*}
&& \Delta_{e_le_h}^{e_he_h} Sh_1(t_lt_l) \leq \Delta_{e_le_l}^{e_he_l} Sh_1(t_lt_l)\\
 &\Rightarrow & \Delta C_{t_l} -\Delta_{e_le_l}^{e_he_l} Sh_1(t_lt_l) \leq \Delta C_{t_l}-\Delta_{e_le_h}^{e_he_h} Sh_1(t_lt_l)\\
&\Rightarrow & \dfrac{\Delta C_{t_l} -\Delta_{e_le_l}^{e_he_l} Sh_1(t_lt_l)}{\Delta_{e_le_h}^{e_he_h} Sh_1(t_lt_h)-\Delta_{e_le_l}^{e_he_l} Sh_1(t_lt_l)}  \leq \dfrac{\Delta C_{t_l}-\Delta_{e_le_h}^{e_he_h} Sh_1(t_lt_l)}{\Delta_{e_le_h}^{e_he_h} Sh_1(t_lt_h)-\Delta_{e_le_h}^{e_he_h} Sh_1(t_lt_l)}
\end{eqnarray*}
Similarly it can be shown that, $\dfrac{\Delta C_{t_h}-\Delta_{e_le_l}^{e_he_l} Sh_1(t_ht_l)}{\Delta_{e_le_l}^{e_he_l} Sh_1(t_ht_h)-\Delta_{e_le_l}^{e_he_l} Sh_1(t_ht_l)} \leq \dfrac{\Delta C_{t_h} -\Delta_{e_le_l}^{e_he_l} Sh_1(t_ht_l)}{\Delta_{e_le_h}^{e_he_h} Sh_1(t_ht_h)-\Delta_{e_le_l}^{e_he_l} Sh_1(t_ht_l)}$. We conclude that, the three intervals mentioned above are pairwise disjoint consequently the strategies $(s_{e_le_l},s_{e_le_l}),(s_{e_le_h},s_{e_le_h}),(s_{e_he_h},s_{e_he_h})$ are rationalizable. 
\end{document}